\documentclass[11pt]{article}
\usepackage[T1]{fontenc}
\usepackage{arcs}
\usepackage{graphicx}
\usepackage{subcaption}
\usepackage{todonotes}
\usepackage{mathtools}
\usepackage{soul,comment,cite}
\usepackage{amsthm}
\usepackage{amsmath}
\usepackage{amssymb}
\usepackage{systeme}

\usepackage{gensymb}
\usepackage[ruled,vlined]{algorithm2e}
\SetKwComment{Comment}{/* }{ */}
\SetKwInOut{Input}{input}\SetKwInOut{Output}{output}
\newtheorem{theorem}{Theorem}%[section]
\newtheorem{corollary}{Corollary}
\newtheorem{lemma}{Lemma}

\graphicspath{{./figures/}}
\usepackage[absolute]{textpos}

\usepackage{fullpage}
\usepackage[top=0.8in, bottom=0.9in, left=0.9in, right=0.9in]{geometry}

\usepackage[pdftex, plainpages = false, pdfpagelabels, 
                 bookmarks=false,
                 bookmarksopen = true,
                 bookmarksnumbered = true,
                 breaklinks = true,
                 linktocpage,
                 pagebackref,
                 colorlinks = true,  % was true
                 linkcolor = blue,
                 urlcolor  = cyan,
                 citecolor = red,
                 anchorcolor = green,
                 hyperindex = true,
                 hyperfigures
                 ]{hyperref}

\def\calD{\mathcal{D}}
\def\bbR{\mathbb{R}}

\def\calh{\mathcal{H}}
\def\vd{\text{Vor}}

\def\dist{\text{dist}}
\title{Maximum Independent Sets in Disk Graphs with Disks in Convex Position\thanks{A preliminary version of this paper will appear in {\em Proceedings of the 20th Scandinavian Symposium on Algorithm Theory (SWAT 2026)}.}}

\author{
Anastasiia Tkachenko\thanks{Kahlert School of Computing,
University of Utah, Salt Lake City, UT 84112, USA. {\tt anastasiia.tkachenko@utah.edu}}
\and
Haitao Wang\thanks{Kahlert School of Computing,
University of Utah, Salt Lake City, UT 84112, USA. {\tt haitao.wang@utah.edu}}
}

\date{}

\begin{document}

\maketitle

\vspace{-0.2in}
\begin{abstract}
For a set $\calD$ of disks in the plane, its disk graph $G(\calD)$ is the graph with vertex set $\calD$, where two vertices are adjacent if and only if the corresponding disks intersect. Given a set $\calD$ of $n$ weighted disks, computing a maximum independent set of $G(\calD)$ is NP-hard. In this paper, we present an $O(n^3\log n)$-time algorithm for this problem in a special setting in which the disks are in \emph{convex position}, meaning that every disk appears on the convex hull of $\calD$. 
This setting has been studied previously for disks of equal radius, for which an $O(n^{37/11})$-time algorithm was known.
Our algorithm also works in the weighted case where disks have weights and the goal is to compute a maximum-weight independent set. 
As an application of our result, we obtain an $O(n^3\log^2 n)$-time algorithm for the \emph{dispersion problem} on a set of $n$ disks in convex position: given an integer $k$, compute a subset of $k$ disks that maximizes the minimum pairwise distance among all disks in the subset.
\end{abstract}


\section{Introduction}
\label{sec:intro}
An \emph{independent set} in a graph is a subset of vertices with no edges between any two of them. The \emph{maximum independent set problem} asks for an independent set of maximum cardinality. In weighted graphs, where each vertex is assigned a weight, the \emph{maximum-weight independent set problem} seeks an independent set of maximum total weight.

In general graphs, these problems are computationally intractable~\cite{ref:KarpRe72}.
This hardness has motivated research along two complementary directions. One direction focuses on designing approximation algorithms, which trade exactness for efficiency. The other aims to identify restricted settings, such as specific graph classes or additional structural constraints, under which exact polynomial-time algorithms become possible. Exact polynomial-time algorithms are known for a variety of restricted graph classes, including circular~\cite{ref:GavrilAl73, ref:ApostolicoNe92}, chordal~\cite{ref:GavrilAl72}, trapezoid~\cite{ref:FelsnerTr97}, outerstring~\cite{ref:BoseCo22}, and Burling graphs~\cite{ref:RzazewskiPo26}, where the additional structure can be exploited algorithmically.

In light of the above, geometric intersection graphs are of particular interest~\cite{ref:AgarwalIn06}. In these graphs, vertices represent geometric objects and edges encode their intersections. In this work, we focus on \emph{disk graphs}, the intersection graphs of disks, which are among the most extensively studied geometric graph families. Beyond their intrinsic theoretical interest, disk graphs provide a natural abstraction for wireless and sensor networks, where connectivity is governed by transmission ranges that are commonly modeled as disks~\cite{ref:ClarkUn90,ref:PerkinsHi94,ref:PerkinsAd99,ref:BalisterCo05}.

Formally, let $\calD$ be a set of disks in the plane. The \emph{disk graph} $G(\calD)$ is the graph with vertex set $\calD$, where two vertices are adjacent if and only if the corresponding disks intersect. Thus, a subset $\calD' \subseteq \calD$ is an independent set in $G(\calD)$ if and only if the disks in $\calD'$ are pairwise disjoint. In addition, each disk is assigned a weight.

The maximum independent set problem in disk graphs remains NP-hard even when all disks have the same radius, in which case the graph is a \emph{unit-disk graph}~\cite{ref:ClarkUn90}. Many approximation algorithms for this problem have been developed; see, for example,~\cite{ref:DasDi20,ref:DasAp15,ref:MaratheSi95,ref:MatsuiAp98,ref:ChanAp12,ref:AgarwalIn06}. Beyond their algorithmic interest, independent sets in intersection graphs arise in several applications, such as map labeling in computational cartography~\cite{ref:AgarwalLa98}. 
For instance, in map labeling one is given a collection of candidate label regions and seeks a largest subset that can be placed without overlap; selecting such a subset is precisely a maximum independent set problem in the corresponding intersection graph.

These results naturally raise the question of whether, and to what extent, geometric structure can impose sufficient constraints to render otherwise intractable combinatorial problems tractable.

\subsection{Our result}

In this paper, we consider a special setting in which the disks in $\calD$ are in \emph{convex position}, meaning that every disk of $\calD$ contributes to the boundary of the convex hull of $\calD$ (see Figure~\ref{fig:convexpos}). Under this assumption, we present an exact algorithm that computes a maximum-weight independent set in $G(\calD)$ in $O(n^3\log n)$ time. This setting bridges the gap between the general two-dimensional problem and one-dimensional variants such as circular graphs, which admit near-linear-time algorithms~\cite{ref:ApostolicoNe92}.

\begin{figure}[t]
\begin{minipage}[h]{0.49\textwidth}
\begin{center}
\includegraphics[height=1.8in]{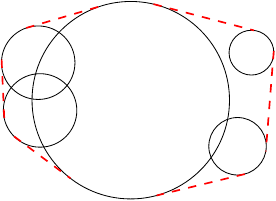}
\caption{\footnotesize Illustrating a set of disks in convex position. The red dashed segments are on the boundary of the convex hull of all disks.}
\label{fig:convexpos}
\end{center}
\end{minipage}
\hspace{0.05in}
\begin{minipage}[h]{0.49\textwidth}
\begin{center}
\includegraphics[height=1.8in]{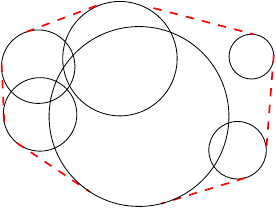}
\caption{\footnotesize Illustrating a set of disks in strongly convex position. The red dashed segments are on the boundary of the convex hull of all disks.}
\label{fig:strongconvexpos}
\end{center}
\end{minipage}
\vspace{-0.1in}
\end{figure}

The maximum-weight independent set problem under this convex position assumption has been studied previously, but only for the unit-disk case in which all disks have the same radius~\cite{ref:SingireddyAl23,ref:TkachenkoDo25}; in that setting, the previously best  algorithm runs in $O(n^{37/11})$ time~\cite{ref:TkachenkoDo25}. Our algorithm applies to disks of arbitrary radii and, moreover, improves the running time even in the unit-disk case.

We say that a set of disks is in \emph{strongly convex position} if every disk has a single maximal disk arc appearing on the boundary of the convex hull of all disks (see Figure~\ref{fig:strongconvexpos}; but the disks in Figure~\ref{fig:convexpos} are not in strongly convex position). We solve the problem gradually as follows.

\begin{enumerate}
    \item
    We first consider a very special case in which the disks of $\calD$ are in strongly convex position and $\calD$ has a maximum-weight independent set that is itself strongly convex; we refer to this as the \emph{strongly-convex-input-and-output} case. Our algorithm for this case is based on dynamic programming and can be viewed as an extension of the unit-disk algorithm of~\cite{ref:TkachenkoDo25} to disks of arbitrary radii. Moreover, we identify a new monotonicity property that allows our algorithm to run faster than the previous unit-disk algorithm~\cite{ref:TkachenkoDo25}. Our algorithm for this case is presented in Section~\ref{sec:inputoutput}. 

    \item
    We then handle the \emph{strongly-convex-input} case, in which $\calD$ is strongly convex but does not necessarily admit a strongly convex maximum-weight independent set. We reduce this case to the first one by adding a set $B$ of $O(n)$ auxiliary points, treated as special disks, such that (i) these special disks appear in any maximum-weight independent set of $\calD \cup B$, and (ii) for any subset $\calD' \subseteq \calD$, the set $\calD' \cup B$ is strongly convex. This case is discussed in Section~\ref{sec:input}.     

    \item
    Finally, we address the most general case in Section~\ref{sec:general}, where $\calD$ is not necessarily strongly convex and does not necessarily admit a strongly convex maximum-weight independent set. We reduce this case to the second one by again adding a set of auxiliary points as special disks. 
\end{enumerate}

As an application of our independent set algorithm, we also solve a \emph{dispersion problem} for a set of $n$ disks in convex position: given an integer $k$, the goal is to find a subset of $k$ disks that maximizes the minimum pairwise distance among all disks in the subset, where the distance between two disks is defined as the minimum distance between any point in one disk and any point in the other. Our algorithm for this problem runs in $O(n^3\log^2 n)$ time and is presented in Section~\ref{sec:dispersion}. 

\subsection{Related work}

Recent developments in geometric algorithms have demonstrated that the convex position constraint can serve as a powerful structural assumption. For a variety of classical problems, it enables exact polynomial-time algorithms, often by exploiting the inherent cyclic structure induced by convexity.

For example, the $k$-center problem is NP-hard for arbitrary point sets, but becomes polynomial-time solvable when the points are in convex position~\cite{ref:ChoiEf23,ref:TkachenkoDo25}. Another recent result that exploits convex position concerns dominating sets in disk graphs~\cite{ref:TkachenkoDo25, ref:TkachenkoCo26}.
Note that for unit-disk graphs, the disk set is in convex position if and only if the centers of the disks are in convex position; this equivalence no longer holds when disks have different radii.

In the literature, related convexity assumptions have also been studied for disk graphs, including settings in which the disk centers are in convex position and settings in which the disks themselves are in convex position. For instance, Herrera and Pérez-Lantero~\cite{ref:HerreraIn21} and Huemer and Pérez-Lantero~\cite{ref:HuemerTh20} investigate combinatorial properties arising when the centers follow a prescribed convex pattern: each disk uses a distinct side of a convex polygon as its diameter, and its center is the midpoint of that side. Under the convex position assumption on the disks, it has also been reported that a maximum clique in the corresponding disk graph can be computed in polynomial time~\cite{ref:ÇağırıcıMa18}.

The convex position assumption has also been explored for classical problems that are already polynomial-time solvable in general, as it can simplify the structure and lead to faster algorithms. A well-known example is the linear-time algorithm for computing the Voronoi diagram of points in convex position~\cite{ref:AggarwalA89}. Additional results for points in convex position can be found in~\cite{ref:LingasOn86,ref:RichardsA90,ref:ChazelleIm93,ref:TkachenkoCo25,ref:HerreraIn21}.

A closely related problem is \emph{dispersion} (also known as the \emph{maximally separated set} problem~\cite{ref:AgarwalCo06}). Given a set of points in the plane and an integer $k$, the goal is to select $k$ points that maximize their minimum pairwise distance; the problem is NP-hard in general~\cite{ref:WangAs88}. When the points are in convex position, however, the problem becomes polynomial-time solvable~\cite{ref:SingireddyAl23,ref:TkachenkoDo25}, and the previously best algorithm runs in $O(n^{37/11}\log n)$ time~\cite{ref:TkachenkoDo25}. Our dispersion problem generalizes this setting from points to disks and our algorithm is even faster.

\section{Preliminaries}
\label{sec:pre}

In this section, we introduce notation and basic concepts used throughout the paper.

Let $\calD$ denote a set of $n$ disks $D_i$ in the plane $\bbR^2$, for $1\le i\le n$. For each disk $D_i$, let $p_i$ be its center and let $r_i$ be its radius. For any subset $\calD' \subseteq \calD$, we use $\calh(\calD')$ to denote the convex hull of $\calD'$.

For any region $R$ in the plane, let $\partial R$ denote its boundary (e.g., $\partial \calh(\calD)$ is the boundary of $\calh(\calD)$ and $\partial D$ is the boundary of a disk $D$). We say that the disks of $\calD$ are in \emph{convex position} if the boundary of every disk appears on $\partial \calh(\calD)$, and they are in \emph{strongly convex position} if the boundary of every disk has a single (maximal) arc appearing on $\partial \calh(\calD)$. For simplicity, for any subset $\calD' \subseteq \calD$, if the disks of $\calD'$ are in (strongly) convex position, we also say that $\calD'$ is (strongly) convex.

For any two points $p$ and $q$ in the plane, let $\overline{pq}$ denote the line segment connecting them, and let $|pq|$ denote the Euclidean distance between $p$ and $q$. For a point $p$ with a \emph{weight} $w(p)$, define
\[
d_p(q) = |pq| + w(p)
\]
for any point $q$ in the plane. We call $d_p(q)$ the \emph{weighted distance} between $p$ and $q$.

For each disk $D_i \in \calD$, we define the weight of its center $p_i$ to be $-r_i$, i.e., the negative of the radius of $D_i$. Hence, for any point $q \in \bbR^2$, we have
$d_{p_i}(q) = |p_i q| - r_i$.
If $q$ lies outside $D_i$, then $d_{p_i}(q)$ equals the minimum Euclidean distance between $q$ and $D_i$. For this reason, we also refer to $d_{p_i}(q)$ as the \emph{distance} between $q$ and $D_i$ (noting that this distance is negative when $q$ lies in the interior of $D_i$).

Two disks $D_i, D_j \in \calD$ are said to be \emph{disjoint} if $|p_i p_j| > r_i + r_j$. A subset $\calD' \subseteq \calD$ is an \emph{independent set} if the disks in $\calD'$ are pairwise disjoint.

Throughout the paper, we assume that the disks of $\calD$ are in convex position (but not necessarily in strongly convex position), and that each disk $D_i$ is assigned a weight $W_i$. Our goal is to compute a maximum-weight independent set of $\calD$, that is, an independent set of disks whose total weight is maximized. Since disks with non-positive weight can be ignored, we assume $W_i > 0$ for all $D_i$. For any subset $\calD' \subseteq \calD$, we use $W(\calD')$ to denote the sum of the weights of the disks in $\calD'$.

For any disk $D$ and any two points $a,b \in \partial D$, we use $\partial_{[a,b]} D$ to denote the portion of $\partial D$ from $a$ to $b$ in counterclockwise order. Similarly, for any subset $\calD' \subseteq \calD$ and any two points $a,b \in \partial \calh(\calD')$, we use $\partial_{[a,b]} \calh(\calD')$ to denote the portion of $\partial \calh(\calD')$ from $a$ to $b$ in counterclockwise order.

Finally, for ease of exposition, we assume general position: no point in the plane is equidistant from four disks of $\calD$, and no line is tangent to three disks.

\section{The strongly-convex-input-and-output case}
\label{sec:inputoutput}

We consider the \emph{strongly-convex-input-and-output} case, in which $\calD$ is strongly convex and admits a maximum-weight independent set that is itself strongly convex. To simplify the discussion, we assume that each disk in $\calD$ has positive radius (i.e., no disk degenerates to a point). The same algorithm applies when some disks are points, but handling such degeneracies would require more technical discussion. We present an algorithm that computes a maximum-weight independent set in $O(n^3\log n)$ time.

We begin with the following lemma.

\begin{lemma}
\label{lem:3IS}
For any subset of three disks $\{D_i, D_j, D_k\} \subseteq \calD$ that forms a strongly convex independent set, there exists a unique point $v_{ijk} \in \bbR^2$ that is equidistant from the three disks. Furthermore, $v_{ijk}$ lies outside each of the three disks.
\end{lemma}

\begin{proof}
Let $\calD' = \{D_i, D_j, D_k\}$. Recall that the center of each disk in $\calD'$ is assigned a weight. We consider the additively weighted Voronoi diagram $\vd(\calD')$ of the centers of $\calD'$, which coincides with the Voronoi diagram of the disks in $\calD'$~\cite{ref:SharirIn85}. Each Voronoi vertex of $\vd(\calD')$ corresponds to a point that is equidistant from the three disks. Thus, it suffices to show that $\vd(\calD')$ has exactly one Voronoi vertex.

Sharir~\cite{ref:SharirIn85} showed that $\vd(\calD')$ has at most two Voronoi vertices. We first argue that $\vd(\calD')$ cannot have two Voronoi vertices. Suppose, for contradiction, that it does. Then the Voronoi cell of at least one disk in $\calD'$ must be closed~\cite{ref:SharirIn85}.

On the other hand, since $\calD'$ is convex, and by the general position assumption, each disk $D \in \calD'$ has at least one point $p$ that lies on $\partial\calh(\calD')$ and does not belong to any other disk in $\calD'$. Let $\rho$ be a ray emanating from $p$ in the direction away from the center of $D$. For any point $q \in \rho$, it is easy to see that $q$ is closer to $D$ than to any other disk in $\calD'$. Hence, $\rho$ is contained in the Voronoi cell of $D$, implying that this Voronoi cell is unbounded and therefore not closed.

This contradicts the assumption that some Voronoi cell in $\vd(\calD')$ is closed. Thus, $\vd(\calD')$ has at most one Voronoi vertex.

We next show that $\vd(\calD')$ has at least one Voronoi vertex. Since $\calD'$ is strongly convex, each disk in $\calD'$ has a single arc appearing on $\partial \calh(\calD')$. This implies that every pair of disks in $\calD'$ is adjacent along $\partial \calh(\calD')$; that is, their arcs on $\partial \calh(\calD')$ are connected by an edge of $\partial \calh(\calD')$ that is a common tangent to the two disks. Consequently, each pair of disks defines a Voronoi edge in $\vd(\calD')$~\cite{ref:SharirIn85}. These three Voronoi edges must intersect at a common point, which is therefore a Voronoi vertex of $\vd(\calD')$.

We conclude that $\vd(\calD')$ has exactly one Voronoi vertex $v_{ijk}$. Since $\calD'$ is an independent set, the disks in $\calD'$ are pairwise disjoint. As $v_{ijk}$ is equidistant from all three disks, it cannot lie inside any of them. This completes the proof.
\end{proof}

In light of the above lemma, for any subset of three disks $\{D_i, D_j, D_k\} \subseteq \calD$ that forms a strongly convex independent set, there exists a disk centered at $v_{ijk}$ that is tangent to all three disks. We use $D_{ijk}$ to denote this disk, and let $r_{ijk}$ denote its radius. Note that $r_{ijk} = d_{p_i}(v_{ijk}) = d_{p_j}(v_{ijk}) = d_{p_k}(v_{ijk})$.

Since $\calD$ is strongly convex, each disk in $\calD$ contributes a single arc to the convex hull boundary $\partial \calh(\calD)$, and $\partial \calh(\calD)$ consists of an alternating sequence of disk arcs and line segments, where each line segment is an outer common tangent of two adjacent disks along $\partial \calh(\calD)$. Let
\[
\calD = \langle D_1, D_2, \ldots, D_n \rangle
\]
be a cyclic list of the disks whose arcs appear on $\partial \calh(\calD)$, ordered counterclockwise. For any $i \neq j$, let $\calD(i,j)$ denote the sublist of $\calD$ from $D_i$ to $D_j$ in counterclockwise order along $\partial \calh(\calD)$ excluding $D_i$ and $D_j$. That is,
\[
\calD(i,j) =
\begin{cases}
\langle D_{i+1}, D_{i+2}, \ldots, D_{j-1} \rangle, & \text{if } i < j, \\
\langle D_{i+1}, D_{i+2}, \ldots, D_n, D_1, \ldots, D_{j-1} \rangle, & \text{if } i > j .
\end{cases}
\]

In addition, we have the following lemma, which will be used later.

\begin{lemma}
\label{lem:orderconsistent}
For any subset $\calD' \subseteq \calD$ that is strongly convex, the counterclockwise order of the disks of $\calD'$ along $\partial \calh(\calD')$ is consistent with the index order of $\calD$.
\end{lemma}

\begin{proof}
Consider a strongly convex subset $\calD' \subseteq \calD$, and let $o$ be a point in the interior of $\calh(\calD')$. Since $\calD' \subseteq \calD$, we have $\calh(\calD') \subseteq \calh(\calD)$, and thus $o$ also lies in the interior of $\calh(\calD)$.

For each disk $D_i \in \calD$, let $q_i$ be a point on $\partial D_i$ that appears on $\partial \calh(\calD)$. Because each disk of $\calD$ contributes a single arc to $\partial \calh(\calD)$, the counterclockwise order of the points $q_i$ along $\partial \calh(\calD)$ is exactly the index order of $\calD$. Moreover, since $\calh(\calD)$ is convex and $o$ lies in its interior, this order coincides with the counterclockwise order of the points $q_i$ with respect to $o$. Hence, the counterclockwise order of the points $q_i$ around $o$ is exactly the index order of $\calD$.

Now consider disks in $\calD'$. For each $D_i \in \calD'$, since $q_i\in \partial \calh(\calD)$ and $\calD'\subseteq\calD$, the point $q_i$ lies on $\partial \calh(\calD')$. Since $\calD'$ is strongly convex, $q_i$ belongs to the unique arc of $\partial D_i$ that appears on $\partial \calh(\calD')$. Therefore, the counterclockwise order of the disks of $\calD'$ along $\partial \calh(\calD')$ coincides with the counterclockwise order of the corresponding points $q_i$ on $\partial \calh(\calD')$. As $\calh(\calD')$ is convex and $o$ lies in its interior, this order is again the counterclockwise order of the points $q_i$ with respect to $o$, which is consistent with the index order of $\calD$. The lemma follows.
\end{proof}

Our algorithm is based on dynamic programming. In Section~\ref{sec:description}, we describe the algorithm, define the subproblems of the dynamic program, and present the dependency relations. We prove the correctness of the algorithm in Section~\ref{sec:correct}, and discuss how to implement it efficiently in Section~\ref{sec:impl}.

\subsection{Algorithm description}
\label{sec:description}

For three disk indices $i,j,k$, we call an ordered triple $(i,j,k)$ a \emph{canonical triple} if $D_i$, $D_j$, and $D_k$ appear on $\partial \calh(\calD)$ in counterclockwise order and $\{D_i, D_j, D_k\}$ is a strongly convex independent set.

For a canonical triple $(i,j,k)$, since $\{D_i, D_j, D_k\}$ is a strongly convex independent set, the disk $D_{ijk}$ exists by Lemma~\ref{lem:3IS}. We define $\calD_k(i,j)$ as the set of disks $D \in \calD(i,j)$ such that $D$ is disjoint from $D_i$, $D_j$, and $D_{ijk}$. For convenience, for any two disk indices $i$ and $j$ such that $D_i$ is disjoint from $D_j$, we also treat $(i,j,0)$ as a canonical triple, and define $\calD_0(i,j)$ to be the set of disks $D \in \calD(i,j)$ that are disjoint from both $D_i$ and $D_j$.

For a canonical triple $(i,j,k)$, including the case $k=0$, let $\calD'_k(i,j)$ denote the subset of disks $D \in \calD_k(i,j)$ such that $\{D, D_i, D_j\}$ is strongly convex.

\paragraph{Subproblems.}
For a canonical triple $(i,j,k)$, including the case $k=0$, we define $f(i,j,k)$ to be the maximum total weight of any subset $\calD' \subseteq \calD_k(i,j)$ such that $\calD' \cup \{D_i, D_j\}$ forms a strongly convex independent set. If no such subset $\calD'$ exists, we set $f(i,j,k) = 0$.

The following lemma gives the dependency relation among the subproblems. Since the proof is lengthy and technical, we defer it to Section~\ref{sec:correct}.

\begin{lemma}
\label{lem:relation}
For any canonical triple $(i,j,k)$, including the case $k=0$, the following holds:
\begin{equation}
\label{eq:deprel}
f(i,j,k) =
\begin{cases}
\displaystyle
\max_{D_l \in \calD'_k(i,j)} \bigl( f(i,l,j) + f(l,j,i) + W_l \bigr), & \text{if } \calD'_k(i,j) \neq \emptyset, \\[1ex]
0, & \text{otherwise.}
\end{cases}
\end{equation}
\end{lemma}

We define $W^*$ to be the optimal objective value of our problem, that is, $W^*$ equals the total weight of a maximum-weight independent set of $\calD$. For convenience, if disks $D_i$ and $D_j$ intersect, we define
$f(i,j,0) = -W_i - W_j$.
The next lemma shows that the optimal value $W^*$ can be computed using the values $f(i,j,0)$ for all pairs $(i,j)$. This explains our choice of $f(i,j,k)$ as the subproblems in the dynamic program.

\begin{lemma}\label{lem:optwgt}
$W^*=\max_{1\leq i,j\leq n} (f(i,j,0)+W_i+W_j)$.
\end{lemma}

\begin{proof}
We first prove that $\max_{1\leq i,j\leq n} \bigl(f(i,j,0)+W_i+W_j\bigr)\leq W^*$. To this end, it suffices to show that $f(i,j,0)+W_i+W_j \le W^*$ holds for any pair $(i,j)$ with $1\leq i,j\leq n$. 

If $D_i$ and $D_j$ intersect, then by definition $f(i,j,0)=-(W_i+W_j)$, and thus $f(i,j,0)+W_i+W_j=0\le W^*$, as claimed (note that $W^*\geq 0$ holds as all disk weights are positive). We now assume that $D_i$ and $D_j$ do not intersect. Let $\calD'$ be a maximum-weight subset of $\calD(i,j)$ such that $\calD'\cup\{D_i,D_j\}$ is a strongly convex independent set.
By definition, $f(i,j,0)=W(\calD')$. Since $W^*$ is the total weight of a maximum-weight independent set of $\calD$ and $\calD$ has a maximum-weight independent set that is strongly convex, the total weight of 
$\calD'\cup\{D_i,D_j\}$ is at most $W^*$. Therefore, it follows that
\[
f(i,j,0)+W_i+W_j
= W(\calD')+W_i+W_j
= W\bigl(\calD'\cup\{D_i,D_j\}\bigr)
\le W^*.
\]

To prove the lemma, it remains to show that $W^* \le \max_{1\leq i,j\leq n} \bigl(f(i,j,0)+W_i+W_j\bigr)$, as follows.

Let $S$ be a maximum-weight independent set of $\calD$ that is strongly convex. By definition, $W(S)=W^*$.
Because $S$ is strongly convex, each disk of $S$ has a single arc appearing on $\partial\calh(S)$. Let $D_i,D_j\in S$ be two disks such that their arcs on $\partial \calh(S)$ are neighboring. Since $S$ is strongly convex, by Lemma~\ref{lem:orderconsistent}, the counterclockwise order of the disks of $S$ appearing on $\partial \calh(S)$ is consistent with the index order of $\calD$. Hence, the disks of $S'=S\setminus\{D_i,D_j\}$ are either all in $\calD(i,j)$ or all in $\calD(j,i)$. Without loss of generality, assume the former. Then, $S'$ is a subset of $\calD_k(i,j)$ with $k=0$ such that $S'\cup \{D_i,D_j\}$ is a strongly convex independent set. By the definition of $f(i,j,0)$, it holds that $W(S')\leq f(i,j,0)$. Therefore, we obtain 
\[W^*=W(S)=W(S')+W_i+W_j\leq f(i,j,0)+W_i+W_j\leq \max_{1\leq i,j\leq n} \bigl(f(i,j,0)+W_i+W_j\bigr).\] 

This completes the proof of the lemma. 
\end{proof}

\subsection{Algorithm correctness: Proving Lemma~\ref{lem:relation}}
\label{sec:correct}

We now prove Lemma~\ref{lem:relation}. Consider a canonical triple $(i,j,k)$. Without loss of generality, we assume that the centers $p_i$ and $p_j$ of disks $D_i$ and $D_j$ lie on the same horizontal line, with $p_i$ to the left of $p_j$. This assumption allows us to clearly distinguish the upper and lower (outer) common tangents of $D_i$ and $D_j$.

If $k \ne 0$, then by definition $\{D_i, D_j, D_k\}$ is a strongly convex independent set, and the disk $D_{ijk}$ exists by Lemma~\ref{lem:3IS}. In the special case $k = 0$, we let $D_{ijk}$ denote the upper halfplane bounded by the upper common tangent line of $D_i$ and $D_j$. Accordingly, its center $v_{ijk}$ is interpreted as a point with infinitely large $y$-coordinate and its radius $r_{ijk}$ is taken to be infinite. With this convention, the arguments below apply uniformly to both cases $k = 0$ and $k \ne 0$.

If $\calD'_k(i,j) = \emptyset$, then $f(i,j,k) = 0$ by Corollary~\ref{coro:zerof}, which is proved using Lemma~\ref{lem:exist}.

\begin{lemma}
\label{lem:exist}
For any subset $\calD' \subseteq \calD_k(i,j)$ such that $\calD' \cup \{D_i, D_j\}$ is a strongly convex independent set, the set $\calD'$ must contain a disk $D$ for which $\{D, D_i, D_j\}$ is a strongly convex independent set.
\end{lemma}

\begin{proof}
Let $S = \calD' \cup \{D_i, D_j\}$. Since $S$ is strongly convex, by Lemma~\ref{lem:orderconsistent}, the counterclockwise order of the disks of $S$ along $\partial \calh(S)$ is consistent with the index order of $\calD$. Because $\calD' \subseteq \calD_k(i,j) \subseteq \calD(i,j)$, the arcs of $D_i$ and $D_j$ on $\partial \calh(S)$ must be adjacent; that is, they are connected by a line segment on $\partial \calh(S)$, and this segment is a common tangent of $D_i$ and $D_j$. Since we have assumed that the centers $p_i$ and $p_j$ lie on a horizontal line, this segment must be the upper common tangent of $D_i$ and $D_j$.

\begin{figure}[t]
\begin{minipage}{\textwidth}
\begin{center}
\includegraphics[height=1.1in]{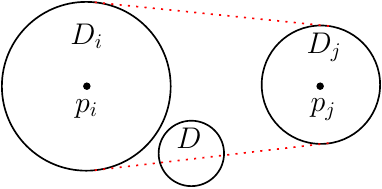}
\caption{\footnotesize Illustrating the proof of Lemma~\ref{lem:exist}.}
\label{fig:exist}
\end{center}
\end{minipage}
\vspace{-0.1in}
\end{figure}

The above argument shows that the upper common tangent of $D_i$ and $D_j$ must be an edge of $\calh(S)$. Since $S$ is strongly convex, the lower common tangent of $D_i$ and $D_j$ cannot be an edge of $\calh(S)$ (since otherwise it would not be possible that each disk of $S$ has a single arc appearing on $\partial \calh(S)$, contradicting that $S$ is strongly convex). This implies that at least one disk $D \in \calD'$ must have a point lying strictly below the supporting line of the lower common tangent of $D_i$ and $D_j$ (see Figure~\ref{fig:exist}). Consequently, the set $\{D, D_i, D_j\}$ is strongly convex. Since $S$ is an independent set, $\{D, D_i, D_j\}$ is also an independent set, completing the proof.
\end{proof}

\begin{corollary}
\label{coro:zerof}
If $\calD'_k(i,j) = \emptyset$, then $\calD_k(i,j)$ does not contain any subset $\calD'$ such that $\calD' \cup \{D_i, D_j\}$ is a strongly convex independent set. Consequently, $f(i,j,k) = 0$.
\end{corollary}

\begin{proof}
Suppose, for contradiction, that $\calD_k(i,j)$ contains a subset $\calD'$ such that $\calD' \cup \{D_i, D_j\}$ is a strongly convex independent set. By Lemma~\ref{lem:exist}, the set $\calD'$ must contain a disk $D$ such that $\{D, D_i, D_j\}$ is a strongly convex independent set. Since $\calD' \subseteq \calD_k(i,j)$, it follows that $\calD_k(i,j)$ contains such a disk $D$, contradicting the assumption that $\calD'_k(i,j) = \emptyset$. The corollary follows.
\end{proof}

In the following, we assume that $\calD'_k(i,j) \neq \emptyset$. We will prove Lemma~\ref{lem:relation} in two parts: a \emph{forward} direction and a \emph{backward} direction.

Let $S$ be an optimal solution defining $f(i,j,k)$; that is, $S \subseteq \calD_k(i,j)$, the set $S \cup \{D_i, D_j\}$ is a strongly convex independent set, and $f(i,j,k) = W(S)$. In the forward direction, we show that $S$ contains a disk $D_l \in \calD'_k(i,j)$ such that $S$ can be partitioned as $S = S_1 \cup S_2 \cup \{D_l\}$, where $S_1 \subseteq \calD_j(i,l)$ and $S_2 \subseteq \calD_i(l,j)$, and both $S_1 \cup \{D_i, D_l\}$ and $S_2 \cup \{D_l, D_j\}$ are strongly convex independent sets. This implies that $f(i,j,k) \le \max_{D_l \in \calD'_k(i,j)} \bigl( f(i,l,j) + f(l,j,i) + W_l \bigr)$.

In the backward direction, we show that for any disk $D_l \in \calD'_k(i,j)$ and any sets $S_1 \subseteq \calD_j(i,l)$ and $S_2 \subseteq \calD_i(l,j)$ such that both $S_1 \cup \{D_i, D_l\}$ and $S_2 \cup \{D_l, D_j\}$ are strongly convex independent sets, the set $S_1 \cup S_2 \cup \{D_l\}$ is contained in $\calD_k(i,j)$ and $S_1 \cup S_2 \cup \{D_l, D_i, D_j\}$ is a strongly convex independent set. This implies that $f(i,j,k) \ge \max_{D_l \in \calD'_k(i,j)} \bigl( f(i,l,j) + f(l,j,i) + W_l \bigr)$.

Together, the two directions establish Lemma~\ref{lem:relation}.

\subsubsection{The forward direction}

Let $S$ be an optimal solution defining $f(i,j,k)$; that is, $S \subseteq \calD_k(i,j)$, the set $S \cup \{D_i, D_j\}$ is a strongly convex independent set, and $f(i,j,k) = W(S)$.

Let $\gamma_{ij}$ denote the \emph{bisector} of $D_i$ and $D_j$, which consists of all points $q$ that are equidistant from $D_i$ and $D_j$. It is well known that $\gamma_{ij}$ is a branch of a hyperbola (degenerating to a line when $r_i = r_j$)~\cite{ref:SharirIn85}. Moreover, since we assume that the centers of $D_i$ and $D_j$ lie on a horizontal line, $\gamma_{ij}$ is $y$-monotone.

\paragraph{A total order $\boldsymbol{\prec}$ on $\boldsymbol{\gamma_{ij}}$.}
For any two points $a,b \in \gamma_{ij}$, we define a total order $a \prec b$ if $a$ has a smaller $y$-coordinate than $b$.
The following lemma will be frequently used.

\begin{figure}[t]
\begin{minipage}[h]{\textwidth}
\begin{center}
\includegraphics[height=1.9in]{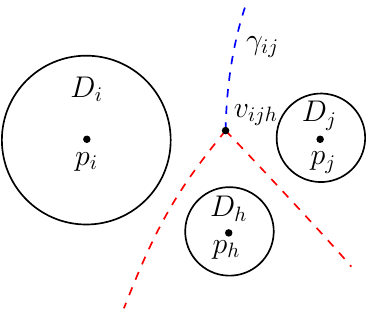}
\caption{\footnotesize Illustrating Lemma~\ref{lem:preorder}. The three dashed curves (two red and one blue) are edges of Voronoi diagram of the three disks. The blue curve belongs to $\gamma_{ij}$.}
\label{fig:voronoi}
\end{center}
\end{minipage}
\vspace{-0.1in}
\end{figure}

\begin{lemma}
\label{lem:preorder}
Suppose $D_h \in \calD$ is a disk such that $\{D_h, D_i, D_j\}$ is a strongly convex independent set, and let $v$ be a point on $\gamma_{ij}$. If $D_h \in \calD(i,j)$, then $v_{ijh} \prec v$ if and only if $d_{p_i}(v) < d_{p_h}(v)$; see Figure~\ref{fig:voronoi}. If $D_h \in \calD(j,i)$, then $v \prec v_{ijh}$ if and only if $d_{p_i}(v) < d_{p_h}(v)$.
\end{lemma}

\begin{proof}
We prove only the case where $D_h \in \calD(i,j)$; the other case can be handled analogously.

Let $\calD' = \{D_h, D_i, D_j\}$. Consider the additively-weighted Voronoi diagram $\vd(\calD')$ of the three weighted centers of $\calD'$. By Lemma~\ref{lem:3IS}, the diagram has a unique Voronoi vertex, namely $v_{ijh}$. Note that $v_{ijh}\in \gamma_{ij}$. Since $\calD'$ is strongly convex, Lemma~\ref{lem:orderconsistent} implies that $D_i$, $D_h$, and $D_j$ appear on $\partial \calh(\calD')$ in counterclockwise order.

Consequently, the portion of the bisector $\gamma_{ij}$ above $v_{ijh}$ belongs to the common boundary between the Voronoi cells of $p_i$ and $p_j$, while the portion of $\gamma_{ij}$ below $v_{ijh}$ lies entirely within the Voronoi cell of $p_h$. Therefore, for any point $v \in \gamma_{ij}$, we have $v_{ijh} \prec v$ if and only if $d_{p_i}(v) < d_{p_h}(v)$.
\end{proof}

By Lemma~\ref{lem:exist}, the set $S$ contains at least one disk $D$ such that $\{D, D_i, D_j\}$ is strongly convex. For each disk $D_h \in S$ with this property, the point $v_{ijh}$ exists by Lemma~\ref{lem:3IS}. Among all such disks $D_h \in S$, let $D_l$ be the one for which $v_{ijl}$ is largest under the order $\prec$. Due to the general position assumption, $D_l$ is unique. Define $S_1 = S \cap \calD(i,l)$ and $S_2 = S \cap \calD(l,j)$. Since $S \subseteq \calD_k(i,j)$, the sets $S_1$, $S_2$, and $\{D_l\}$ form a partition of $S$. By definition, $D_l \in \calD'_k(i,j)$.

To complete the proof of the forward direction, it remains to show that $S_1 \subseteq \calD_j(i,l)$, $S_2 \subseteq \calD_i(l,j)$, and that both $S_1 \cup \{D_i, D_l\}$ and $S_2 \cup \{D_l, D_j\}$ are strongly convex independent sets. These statements are proved in the next two lemmas.

\begin{lemma}\label{lem:notintersect}
    $S_1\subseteq \calD_j(i,l)$ and $S_2\subseteq \calD_i(l,j)$. 
\end{lemma}
\begin{proof}
We only prove $S_1\subseteq \calD_j(i,l)$; the other case can be proved analogously. 

Consider a disk $D_h\in S_1$. Our goal is to prove $D_h\in \calD_j(i,l)$. To this end, we need to show that $D_h$ is disjoint from the three disks $D_i$, $D_l$, and $D_{ijl}$. Since $S_1\subseteq S$, $D_l\in S$, and $S\cup \{D_i,D_j\}$ is an independent set, we know that $D_h$ is disjoint from both $D_i$ and $D_l$. Hence, it remains to prove that $D_h$ is disjoint from $D_{ijl}$, or equivalently, $d_{p_h}(v_{ijl})>r_{ijl}$; recall that $r_{ijl}$ denotes the radius of $D_{ijl}$ and $r_{ijl}=d_{p_i}(v_{ijl})$. Hence, it suffices to show that $d_{p_h}(v_{ijl})>d_{p_i}(v_{ijl})$. Let $S'=\{D_h,D_i,D_j\}$. 
Depending on whether $S'$ is strongly convex, we distinguish two cases. 

If $S'$ is strongly convex, then by Lemma~\ref{lem:3IS}, $v_{ijh}$ exists. 
By the definition of $D_l$, we have $v_{ijh}\prec v_{ijl}$. Consequently, applying Lemma~\ref{lem:preorder} (with $v=v_{ijl}$) obtains 
$d_{p_i}(v_{ijl})<d_{p_h}(v_{ijl})$.

\begin{figure}[t]
\begin{minipage}[h]{\textwidth}
\begin{center}
\includegraphics[height=1.7in]{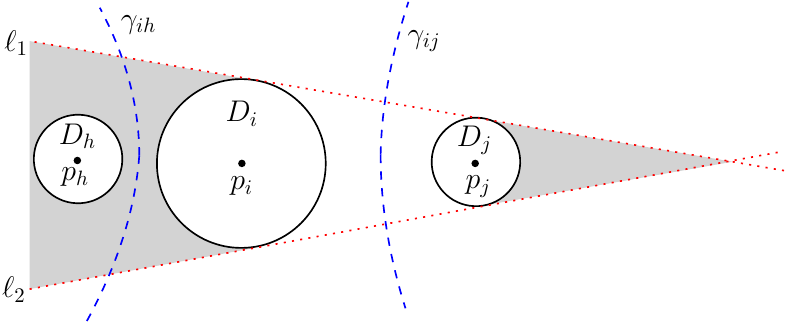}
\caption{\footnotesize The disk $D_h$ can only be in one of the two shaded regions. The curve $\gamma_{ij}$ is the bisector of $D_i$ and $D_j$, and $\gamma_{ih}$ is the bisector of $D_i$ and $D_h$. The two bisectors are the only two edges of the Voronoi diagram of the three disks.}
\label{fig:notintersect}
\end{center}
\end{minipage}
\vspace{-0.1in}
\end{figure}

We now consider the case in which $S'$ is not strongly convex. First of all, since $\calD$ is strongly convex and $S'\subseteq \calD$, $S'$ must be convex. Since $S\cup \{D_i,D_j\}$ is strongly convex and $S\subseteq \calD(i,j)$, as argued in the proof of Lemma~\ref{lem:exist}, the upper common tangent of $D_i$ and $D_j$ must be on $\partial \calh(S\cup \{D_i,D_j\})$. Since $S'\subseteq S\cup \{D_i,D_j\}$, the upper common tangent of $D_i$ and $D_j$ must be on $\partial \calh(S')$. Let $\ell_1$ and $\ell_2$ be the upper and lower common tangent lines of $D_i$ and $D_j$, respectively. 

Since $S'$ is independent and convex but not strongly convex, $D_h$ must be either in the region to the left of $D_i$, below $\ell_1$, and above $\ell_2$, or in the region to the right of $D_j$, below $\ell_1$, and above $\ell_2$ (see Figure~\ref{fig:notintersect}). In either case, if we consider the Voronoi diagram $\vd(S')$ of the disks of $S'$, then the bisector $\gamma_{ij}$ is a common edge of the Voronoi cells of $D_i$ and $D_j$. This implies that for any point $v\in \gamma_{ij}$, it holds that $d_{p_i}(v)<d_{p_h}(v)$. Since $v_{ijl}\in \gamma_{ij}$, we obtain $d_{p_i}(v_{ijl})<d_{p_h}(v_{ijl})$.
\end{proof}

\begin{lemma}\label{lem:pocket}
Both $S_1\cup\{D_i,D_l\}$ and $S_2\cup\{D_l,D_j\}$ are strongly convex independent sets.
\end{lemma}
\begin{proof}
We only prove the case for $S_1\cup\{D_i,D_l\}$; the other case can be proved analogously. 

Let $S'=S\cup \{D_i,D_j\}$. Consider the convex hull $\calh(S')$. 
Since $S'$ is strongly convex, each disk of $S'$ has a single arc appearing on $\partial \calh(S')$. 
We next define several points (see Figure~\ref{fig:pockets}). 
Let $a_i$ be the counterclockwise endpoint of the arc of $\partial D_i$ appearing on $\partial \calh(S')$. 
Let $b_l$ be the clockwise endpoint of the arc of $\partial D_l$ appearing on $\partial \calh(S')$. 
Let $c_i$ and $c_l$ be the points of $\partial D_i$ and $\partial D_l$ tangent to $D_{ijl}$, respectively.

\begin{figure}[t]
\begin{minipage}[h]{\textwidth}
\begin{center}
\includegraphics[height=2.2in]{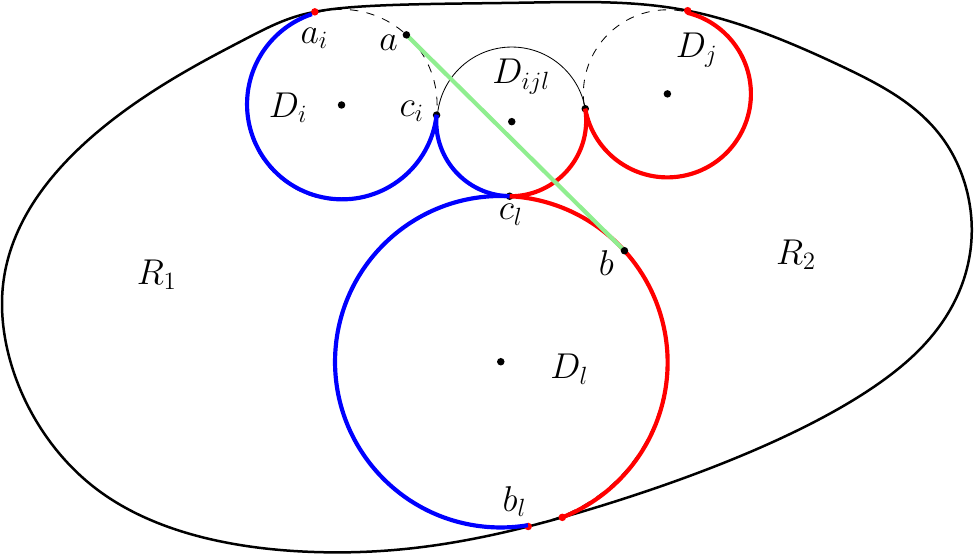}
\caption{\footnotesize Illustrating the proof of Lemma~\ref{lem:pocket}. The blue disk arcs are on $\partial R_1$. The red disk arcs are on $\partial R_2$.}
\label{fig:pockets}
\end{center}
\end{minipage}
\vspace{-0.1in}
\end{figure}

Since $S'$ is strongly convex and $S_1\subseteq \calD(i,l)$, by Lemma~\ref{lem:orderconsistent}, for each disk $D\in S_1$, the arc of $\partial D$ appearing in $\partial \calh(S')$ must be on $\partial_{[a_i,b_l]}\calh(S')$, the portion of $\partial \calh(S')$ counterclockwise from $a_i$ to $b_l$. 

Consider the region $R_1$ bounded by following curves in the counterclockwise order (see Figure~\ref{fig:pockets}): $\partial_{[a_i,b_l]}\calh(S')$, $\partial_{[c_l,b_l]}D_l$, $\partial_{[c_i,c_l]}D_{ijl}$, and $\partial_{[a_i,c_i]}D_i$. 
For each disk $D_h\in S_1$, 
since the arc of $\partial D_h$ appearing on $\partial \calh(S')$ is on $\partial_{[a_i,b_l]}\calh(S')$ and $D_h$ is disjoint from $D_i$, $D_l$, and $D_{ijl}$ by Lemma~\ref{lem:notintersect}, we obtain that $D_h$ must be contained in the region $R_1$. 

Starting from $b_l$, we move a point $q$ counterclockwise on $\partial D_l$ until a point $b$ such that $\ell_q$ is tangent to $D_i$, where $\ell_q$ is the line through $q$ and tangent to $D_l$; see Figure~\ref{fig:pockets}. We assign a direction to $\ell_q$ so that $D_l$ is on the left of $\ell_q$. Note that initially when $q$ is at $b_l$, since $b_l\in \partial \calh(S')$, every disk of $S'$ is on the left side of $\ell_q$. By definition, when $q$ is at $b$, $\ell_q$ is tangent to $D_i$ and $D_i$ is on the left side of $\ell_q$; let $\ell_b$ be $\ell_q$ when $q$ is at $b$. By definition, $\partial_{[b_l,b]}D_l$ must be on $\partial \calh(\{D_l,D_i\})$.

Similarly, starting from $a_i$, we move a point $p$ clockwise on $\partial D_i$ until a point $a$ such that $\ell_p$ is tangent to $D_l$, where $\ell_p$ is the line through $p$ and tangent to $D_i$. We assign a direction to $\ell_p$ so that $D_i$ is on the right of $\ell_p$. Note that initially when $p$ is at $a_i$, since $a_i\in \partial \calh(S')$, every disk of $S'$ is on the right side of $\ell_p$. By definition, when $p$ is at $a$, $\ell_p$ is tangent to $D_l$ and $D_l$ is on the right side of $\ell_p$; let $\ell_a$ be $\ell_p$ when $p$ is at $a$. By definition, $\partial_{[a,a_i]}D_i$ must be on $\partial \calh(\{D_l,D_i\})$. In addition, $\ell_a$ is exactly $\ell_b$, but $\ell_a$ and $\ell_b$ have opposite directions. 

We claim that $c_l\in \partial_{[b,b_l]}D_l$; see Figure~\ref{fig:pockets}. Indeed, by definition, $c_l$ is the intersection of $\overline{p_lv_{ijl}}$ with $\partial D_l$, with $p_l$ as the center of $D_l$ and $v_{ijl}$ as the center of $D_{ijl}$. Since $v_{ijl}$ is equidistant from $D_i$ and $D_l$, $c_l$ cannot be on the boundary of the convex hull of $D_i$ and $D_l$~\cite{ref:SharirIn85}. Since $\partial_{[b_l,b]}D_l\subseteq\partial \calh(\{D_l,D_i\})$, it follows that $c_l\not\in \partial_{[b_l,b]}D_l$. Therefore, $c_l$ must be on $\partial_{[b,b_l]}D_l$.

By a similar argument, we have $c_i\in \partial_{[a_i,a]}D_i$.

Now consider the region $R_1'$ bounded by following curves in the counterclockwise order: $\partial_{[a_i,b_l]}\calh(S')$, $\partial_{[b_l,b]}D_l$, $\overline{ba}$, and $\partial_{[a,a_i]}D_i$. Since $c_l\in \partial_{[b,b_l]}D_l$ and $c_i\in \partial_{[a_i,a]}D_i$, it must be the case that $R_1\subseteq R_1'$. Since all disks of $S_1$ are in $R_1$, they are also in $R_1'$. Based on how we obtain the common tangent $\overline{ab}$ of $D_i$ and $D_l$, observe that if we traverse the above bounding curves of $R_1'$ counterclockwise around $R_1'$, we always make left turn, implying that $R_1'$ is convex. Furthermore, $R_1'$ is exactly $\calh(S_1\cup \{D_i,D_l\})$. By definition, each disk of $S_1\cup \{D_i,D_l\}$ has a single arc appearing on $\partial R_1'$. Indeed, for each disk $D\in S_1$, its single arc on $\partial R_1'$ is in $\partial_{[a_i,b_l]}\calh(S')$. The single arc of $D_i$ on $\partial R_1'$ is $\partial_{[a,a_i]}D_i$ and the single arc of $D_l$ on $\partial R_1'$ is $\partial_{[b_l,b]}D_l$.
This proves that $S_1\cup \{D_i,D_l\}$ is strongly convex. 

We remark that since $R_1'$ is the convex hull of $S_1\cup \{D_i,D_l\}$ and $S_1\cup \{D_i,D_l\}\subseteq S'$, we have $R_1'\subseteq \calh(S')$. As $R_1\subseteq R_1'$, $R_1\subseteq \calh(S')$. We call $R_1$ the {\em left pocket} of $\calh(S')$. Similarly, we can define a right pocket $R_2$ of $\calh(S')$ using the disks of $S_2$ (see Figure~\ref{fig:pockets}). These two pockets are disjoint because they are separated by $D_i$, $D_l$, $D_j$, and $D_{ijl}$. These concepts will be used later. 
\end{proof}

\subsubsection{The backward direction}
Consider a disk $D_l\in\calD'_k(i,j)$ together with  $S_1\subseteq \calD_j(i,l)$ and $S_2\subseteq \calD_i(l,j)$ such that both $S_1\cup \{D_i,D_l\}$ and $S_2\cup \{D_l,D_j\}$ are strongly convex independent sets. Our goal is to show that $S_1\cup S_2\cup\{D_l\}\subseteq \calD_k(i,j)$ and $S_1\cup S_2\cup\{D_l,D_i,D_j\}$ is a strongly convex independent set. Let $S=S_1\cup S_2\cup\{D_l\}$. 

We begin with arguing $S\subseteq \calD_k(i,j)$. Let $D_h$ be a disk of $S$. If $D_h$ is $D_l$, then since $D_l\in \calD'_k(i,j)$ and $\calD'_k(i,j)\subseteq \calD_k(i,j)$, we have $D_l\in \calD_k(i,j)$. Otherwise, $D_h$ is either in $S_1$ or in $S_2$. 
We assume that $D_h\in S_1$ (the other case can be handled analogously). In the following, we argue that $D_h\in \calD_k(i,j)$. To this end, since $D_h\in S_1\subseteq \calD_j(i,l)\subseteq \calD(i,l)\subseteq \calD(i,j)$, we need to show that $D_h$ is disjoint from $D_i$, $D_j$, and $D_{ijk}$. 

First of all, since $D_h\in S_1\subseteq  \calD_j(i,l)$, $D_h$ must be disjoint from $D_i$. In the next two lemmas, we prove that $D_h$ is disjoint from $D_{ijk}$ and $D_j$, respectively. 

\begin{lemma}
$D_h$ is disjoint from $D_{ijk}$. 
\end{lemma}
\begin{proof}
It suffices to show that $d_{p_h}(v_{ijk})>r_{ijk}$, or equivalently, $d_{p_h}(v_{ijk})>d_{p_i}(v_{ijk})$ as $r_{ijk}=d_{p_i}(v_{ijk})$. 

Since $D_l\in \calD_k'(i,j)$, $\{D_l,D_i,D_j\}$ is a strongly convex independent set and thus $v_{ijl}$ exists by Lemma~\ref{lem:3IS}. Clearly, both $v_{ijk}$ and $v_{ijl}$ are on the bisector $\gamma_{ij}$ of $D_i$ and $D_j$. Depending on whether $\{D_i,D_j,D_h\}$ is strongly convex, we distinguish two cases. 

\begin{itemize}
    \item 

If $\{D_i,D_j,D_h\}$ is not strongly convex, then similarly to the proof of Lemma~\ref{lem:notintersect}, we can obtain that for any point $v\in \gamma_{ij}$, $d_{p_i}(v)<d_{p_h}(v)$. As $v_{ijk}\in \gamma_{ij}$, it follows that $d_{p_i}(v_{ijk})<d_{p_h}(v_{ijk})$. 

\item
If $\{D_i,D_j,D_h\}$ is strongly convex, then $v_{ijh}$ exists by Lemma~\ref{lem:3IS}. We claim that $v_{ijh}\prec v_{ijl}$. To see this, since $D_h\in S_1\subseteq \calD_j(i,l)$, $D_h$ is disjoint from $D_{ijl}$, or equivalently, $d_{p_i}(v_{ijl})<d_{p_h}(v_{ijl})$. By Lemma~\ref{lem:preorder} (with $v=v_{ijl}$), we obtain $v_{ijh}\prec v_{ijl}$. 

Since $D_l\in \calD_k(i,j)$, following the similar argument as above, we can obtain $v_{ijl}\prec v_{ijk}$. Hence, we have $v_{ijh}\prec v_{ijk}$. Consequently, applying Lemma~\ref{lem:preorder} again (with $v=v_{ijk}$) leads to $d_{p_i}(v_{ijk})<d_{p_h}(v_{ijk})$. 
\end{itemize}

The lemma thus follows. 
\end{proof}

\begin{lemma}\label{lem:diskdisjoint}
    $D_h$ is disjoint from $D_j$. 
\end{lemma}
\begin{proof}
Recall that $D_h\in S_1\subseteq D_j(i,l)\subseteq D(i,l)$ and $S_1\cup \{D_i,D_l\}$ is a strongly convex independent subset. 
Let $S_1'=S_1\cup \{D_i,D_l\}$ and $\calD'=\{D_i,D_l,D_j\}$. Since $D_l\in D_k'(i,j)$, $\calD'$ is a strongly convex independent set. Note that $D_j\in \calD(l,i)$. 

Let $\calD_1'=S_1'\cup \{D_j\}$. 
We claim that $\calD_1'$ must be strongly convex. Indeed, since $\calD$ is strongly convex, $\partial D_i$ has a point $q_i$ on $\partial \calh(\calD)$. Similarly, $\partial D_l$ has a point $q_l$ on $\partial \calh(\calD)$. Since $S_1'$ is strongly convex, $\partial D_i$ has a single arc $\alpha_i$ that appears on $\partial \calh(S_1')$. Since $\calD'$ is strongly convex, $\partial D_i$ has a single arc $\alpha_i'$ that appears on $\partial \calh(\calD')$. Clearly, $q_i$ is on both $\alpha_i$ and $\alpha_i'$. Because $\calD_1'=S_1'\cup \calD'$, the portion of $\partial D_i$ that appears on $\partial \calh(\calD_1')$ is $\alpha_i\cap \alpha_i'$, which is a single arc since both of them contain $q_i$. Therefore, we obtain that $\partial D_i$ has a single arc appearing on $\partial \calh(\calD_1')$. Similarly, $\partial D_l$ also has a single arc appearing on $\partial \calh(\calD_1')$ and that arc contains $q_l$.

\begin{figure}[t]
\begin{minipage}[h]{\textwidth}
\begin{center}
\includegraphics[height=2.2in]{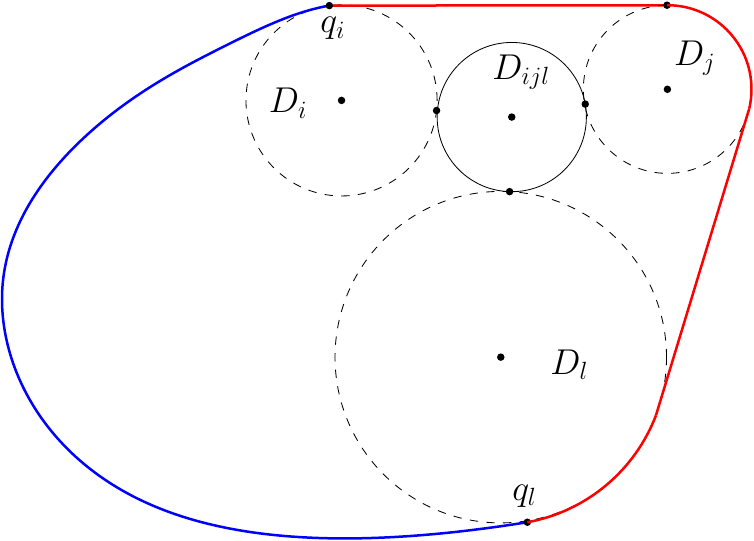}
\caption{\footnotesize Illustrating the proof of Lemma~\ref{lem:diskdisjoint}. The blue curve is $\partial_{[q_i,q_l]}\calh(S_1')$ and the red curve is $\partial_{[q_l,q_i]}\calh(\calD')$. $\partial \calh(\calD_1')$ is exactly the union of the two curves. 
}
\label{fig:pockets00}
\end{center}
\end{minipage}
\vspace{-0.1in}
\end{figure}

Since $S_1'$ is strongly convex, by Lemma~\ref{lem:orderconsistent}, the counterclockwise order of the disks of $S_1'$ appearing on $\partial \calh(S_1')$ is consistent with the index order of $\calD$. Similarly, since $\calD'$ is strongly convex, the order of the disks of $\calD'$ whose arcs appear on $\partial \calh(\calD')$ counterclockwise is consistent with the index order of $\calD$. As $S_1\in \calD(i,l)$ while $D_j\in \calD(l,i)$, we obtain that $\partial_{[q_i,q_l]}\calh(S_1')\cup \partial_{[q_l,q_i]}\calh(\calD')$ is exactly $\partial \calh(\calD_1')$; see Figure~\ref{fig:pockets00}. Furthermore, by definition, each disk of $S_1$ has a single arc appearing on $\partial_{[q_i,q_l]}\calh(S_1')$ and $D_j$ has a single arc appearing on $\partial_{[q_l,q_i]}\calh(\calD')$. Therefore, every disk of $\calD_1'$ has a single arc appearing on $\partial\calh(\calD_1')$. It follows that $\calD_1'$ is strongly convex. 

We consider Voronoi diagram $\vd(\calD_1')$ of the disks of $\calD_1'$. 
Since $S_1\subseteq D_j(i,l)$, no disk of $S_1$ intersects the disk $D_{ijl}$. This implies that the center $v_{ijl}$ of $D_{ijl}$ is a Voronoi vertex of $\vd(\calD_1')$ since it is equidistant from $D_i$, $D_l$, and $D_j$. 
Because $\calD_1'$ is strongly convex, $D_j$ has a single arc appearing on $\partial \calh(\calD_1')$ and its two neighboring arcs on $\partial \calh(\calD_1')$ are from $D_i$ and $D_l$, respectively. Therefore, the bisector $\gamma_{ij}$ of $D_i$ and $D_j$ defines a Voronoi edge with one endpoint at $v_{ijl}$ and the other one at the infinity~\cite{ref:SharirIn85}. Similarly, the bisector $\gamma_{lj}$ of $D_l$ and $D_j$ defines a Voronoi edge with one endpoint at $v_{ijl}$ and the other one at the infinity. This implies that the Voronoi cell $C_j$ of $D_j$ has only two neighboring Voronoi cells, one is defined by $D_i$ and the other defined by $D_l$. 

Now we are ready to show that $D_h$ is disjoint from $D_j$. It suffices to argue that $d_{p_j}(p_h)>r_h$, the radius of $D_h$. According to the above analysis, the Voronoi cell of $D_h$ (or equivalently, the Voronoi cell of the weighted center $p_h$ of $D_h$) is not adjacent to $C_j$. Let $\calD''$ denote the set of disks whose Voronoi cells are adjacent to the Voronoi cell of $D_h$. 
Since $D_j\not\in \calD''$, by the property of Voronoi diagrams, $d_{p_j}(p_h)$ is larger than the minimum weighted distance from $p_h$ to the centers of all disks of $\calD''$. Since $\calD''\subseteq S_1'$, $D_h\in S_1'$, and $S_1'$ is an independent set, $\calD''\cup \{D_h\}$ is also an independent set. Thus, the weighted distance from $p_h$ to every disk center of $\calD''$ is larger than $r_h$. It follows that $d_{p_j}(p_h)>r_h$. 
\end{proof}

The above proves that $S\subseteq \calD_k(i,j)$. 
The following lemma finally proves that $S\cup \{D_i,D_j\}$ is a strongly convex independent set. This completes the proof of the backward direction.

\begin{lemma}\label{lem:convex}
    $S\cup \{D_i,D_j\}$ is a strongly convex independent set. 
\end{lemma}
\begin{proof}
Let $S'=S\cup \{D_i,D_j\}$. 
We first argue that $S'$ is strongly convex, by extending the proof in Lemma~\ref{lem:diskdisjoint}. 

Let $\calD'=\{D_i,D_l,D_j\}$. Let $S_1'=S_1\cup \calD'$ and $S_2'=S_1\cup \calD'$. 
The proof of Lemma~\ref{lem:diskdisjoint} already argues that $S_1'$ is strongly convex. Similarly, $S_2'$ is also strongly convex. 

Since $\calD$ is strongly convex, $\partial D_i$ has a point $q_i$ on $\partial \calh(\calD)$. Similarly, define $q_l$ and $q_j$ on $D_l$ and $D_j$, respectively. Since $S_1'$ is strongly convex, $\partial D_i$ has a single arc $\alpha_i^1$ that appears on $\partial \calh(S_1')$. Since $S_2'$ is strongly convex, $\partial D_i$ has a single arc $\alpha_i^2$ that appears on $\partial \calh(S_2')$. Clearly, $q_i$ is on both $\alpha_i^1$ and $\alpha_i^2$. Because $S'=S_1'\cup S_2'$, the portion of $\partial D_i$ that appears on $\partial \calh(S')$ is $\alpha_i^1\cap \alpha_i^2$, which is a single arc since both of them contain $q_i$. Therefore, we obtain that $\partial D_i$ has a single arc $\alpha_i$ appearing on $\partial \calh(S')$. Similarly, $\partial D_l$ has a single arc $\alpha_l$ appearing on $\partial \calh(S')$ and that arc contains $q_l$, and $\partial D_j$ has a single arc $\alpha_j$ appearing on $\partial \calh(S')$ and that arc contains $q_j$.  

Since $S_1'$ is strongly convex, by Lemma~\ref{lem:orderconsistent}, the order of the disks of $S_1'$ whose arcs appear on $\partial \calh(S_1')$ counterclockwise is consistent with the index order of $\calD$. As $S_1\subseteq \calD(i,l)$, for each disk $D\in S_1$, the single arc of $\partial D$ appearing on $\partial \calh(S_1')$ is on $\partial_{[q_i,q_l]}\calh(S_1')$. Similarly, as $S_2\subseteq \calD(l,j)$, for each disk $D\in S_2$, the single arc of $\partial D$ appearing on $\partial \calh(S_2')$ is on $\partial_{[q_l,q_j]}\calh(S_2')$. 

Due to the convexity of the three arcs $\alpha_i$, $\alpha_l$, and $\alpha_j$, $\partial \calh(S')$ is exactly $\partial_{[q_i,q_l]}\calh(S_1')\cup \partial_{[q_l,q_j]}\calh(S_2')\cup \partial_{[q_j,q_i]}\calh(\{D_i,D_j\})$; see Figure~\ref{fig:pockets10}. For each disk $D\in S'$, it has a single arc appearing on $\partial \calh(S')$. Indeed, if $D\in \calD'$, then we already showed above that this is true. If $D\in S_1$, then $\partial D$ only appears on $\partial_{[q_i,q_l]}\calh(S_1')$ and it has a single arc appearing there. Therefore, $D$ has a single arc appearing on $\partial \calh(S)$. If $D\in S_2$, then the argument is analogous. 
In summary, each disk of $S'$ has a single arc appearing on $\partial \calh(S')$. Hence, $S'$ is strongly convex.

\begin{figure}[t]
\begin{minipage}[h]{\textwidth}
\begin{center}
\includegraphics[height=2.2in]{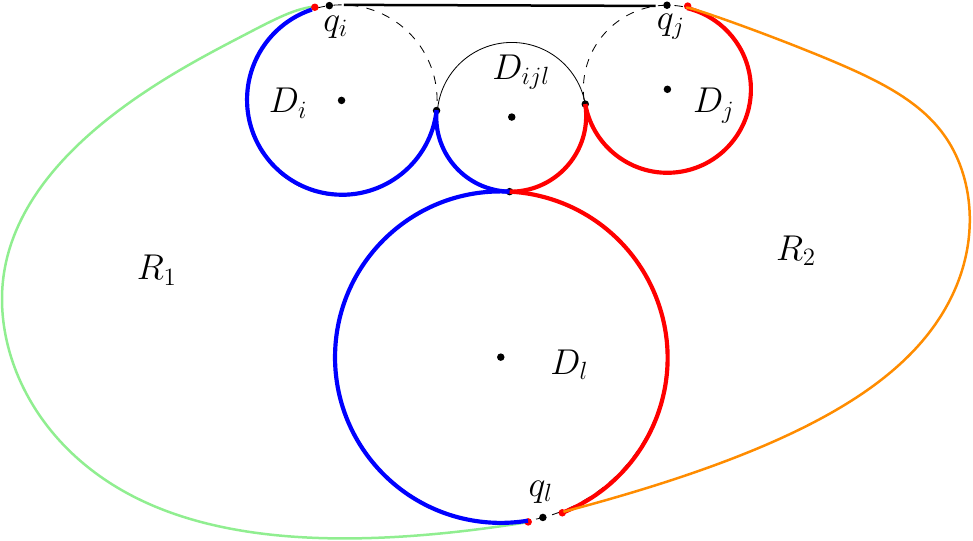}
\caption{\footnotesize Illustrating the proof of Lemma~\ref{lem:convex}. $R_1$ is bounded by the green curve and the three blue disk arcs, and $R_2$ is bounded by the orange curve and the three red disk arcs. All disks of $S_1$ (resp., $S_2$) are in $R_1$ (resp., $R_2$).}
\label{fig:pockets10}
\end{center}
\end{minipage}
\vspace{-0.1in}
\end{figure}

We next prove that $S'$ is an independent set. We already know that $\calD'$ is an independent set. Note also that $S_1'$ is an independent set. Indeed, by definition, $S_1\cup \{D_i,D_l\}$ is an independent set. By Lemma~\ref{lem:diskdisjoint}, no disk of $S_1$ intersects $D_j$. Since $\calD'$ is an independent set and $S_1'=S_1\cup \calD'$, $S_1'$ is an independent. Following a similar argument, $S_2'$ is also an independent set. 
Now to prove that $S'$ is an independent set, 
it remains to show that no disk of $S_1$ can intersect any disk of $S_2$.  

We consider the convex hull $\calh(S')$ of $S'$. Following the same way as in the proof of Lemma~\ref{lem:pocket}, we can define a left pocket $R_1$ and right pocket $R_2$ of $\calh(S')$ (see Figures~\ref{fig:pockets} or \ref{fig:pockets10}). As argued in the proof of Lemma~\ref{lem:pocket}, all disks of $S_1$ are in $R_1$ while all disks of $S_2$ are in $R_2$. Since $R_1$ and $R_2$ are disjoint, no disk of $S_1$ can intersect any disk of $S_2$. It follows that $S'$ is an independent set. 

This completes the proof of the lemma. 
\end{proof}

\subsection{Algorithm implementation}
\label{sec:impl}

We now discuss how to implement the algorithm. We focus on computing $W^*$. By slightly modifying the algorithm using a standard backtracking technique, an actual maximum-weight independent set can also be recovered within asymptotically the same time complexity.

To compute $W^*$, Lemma~\ref{lem:optwgt} implies that it suffices to compute $f(i,j,k)$ for all canonical triples $(i,j,k)$, including the case $k=0$. To this end, by Equation~\eqref{eq:deprel}, we must determine an order for solving the subproblems such that, when computing a value $f(i,j,k)$, all values $f(i,l,j)$ and $f(l,j,i)$ for $D_l \in \calD'_k(i,j)$ have already been computed.

We process pairs $(i,j)$ in increasing order of $j$, and for each fixed $j$ in decreasing order of $i$: for $j=2,\ldots,n$, we consider $i=j-1,j-2,\ldots,1$ in this order. If $D_i \cap D_j \neq \emptyset$, we simply set $f(i,j,0)=-W_i-W_j$. Otherwise, we compute $f(i,j,0)$ using Equation~\eqref{eq:deprel}. Next, for each disk $D_k \in \calD(j,i)$ such that $(i,j,k)$ is canonical, we compute $f(i,j,k)$ using Equation~\eqref{eq:deprel}. With this ordering, every $f$-value appearing on the right-hand side of Equation~\eqref{eq:deprel} involves indices that lie between $i$ and $j$ in cyclic order, and hence has already been computed.

Using Equation~\eqref{eq:deprel}, each subproblem $f(i,j,k)$ can be computed in $O(n)$ time by scanning all disks $D_l \in \calD'_k(i,j)$. Since there are $O(n^3)$ subproblems, the total running time of the algorithm is $O(n^4)$.

\paragraph{An improved solution.}
We now improve the running time to $O(n^3\log n)$. More specifically, we show that for any fixed pair $(i,j)$, the $O(n)$ subproblems $f(i,j,k)$ for all relevant indices $k$ can be computed in a total of $O(n\log n)$ time. This improvement is achieved by exploiting a monotonicity property of the sets $\calD'_k(i,j)$ with respect to a suitable ordering of the indices $k$.

Consider a pair $(i,j)$ such that $D_i$ does not intersect $D_j$, and let $(i,j,k)$ be a canonical triple. For each disk $D_l \in \calD'_k(i,j)$, we define its \emph{cost} as
\[
    \mathit{cost}(D_l)=f(i,l,j)+f(l,j,i)+W_l.
\]
Define $\calD'(i,j)$ as the set of disks $D \in \calD(i,j)$ such that $\{D,D_i,D_j\}$ forms a strongly convex independent set. Then $\calD'_k(i,j)$ consists precisely of those disks in $\calD'(i,j)$ that are disjoint from the disk $D_{ijk}$. By Equation~\eqref{eq:deprel}, computing $f(i,j,k)$ reduces to finding the disk of maximum cost among all disks in $\calD'(i,j)$ that are disjoint from $D_{ijk}$.

We process a fixed pair $(i,j)$ using the following approach, which computes $f(i,j,k)$ for all $k \in K$, where $K$ denotes the set of indices $k$ such that $(i,j,k)$ is a canonical triple. The total running time of this procedure is $O(n\log n)$.

Without loss of generality, we assume that the centers $p_i$ and $p_j$ lie on the same horizontal line, with $p_i$ to the left of $p_j$. We first compute the set $K$ and the set $\calD'(i,j)$, which can be done in $O(n)$ time.

For each $k \in K$, since $(i,j,k)$ is a canonical triple, the set $\{D_i,D_j,D_k\}$ is strongly convex, and thus the vertex $v_{ijk}$ exists by Lemma~\ref{lem:3IS}. The following lemma, which establishes a monotonicity property of the sets $\calD'_k(i,j)$, is central to our approach.

\begin{lemma}\label{lem:monotone}
Consider any $k,k' \in K$ with $v_{ijk} \prec v_{ijk'}$. For any disk $D_h \in \calD'(i,j)$, if $D_h$ is disjoint from $D_{ijk}$, then $D_h$ is also disjoint from $D_{ijk'}$. Consequently, $\calD'_k(i,j) \subseteq \calD'_{k'}(i,j)$.
\end{lemma}
\begin{proof}
    Consider a disk $D_h\in \calD'(i,j)$ such that $D_h$ is disjoint from $D_{ijk}$. Our goal is to show that $D_h$ is also disjoint from $D_{ijk'}$, or equivalently, $d_{p_h}(v_{ijk'})>r_{ijk'}$. 

    Since $D_h$ is disjoint from $D_{ijk}$, we have $d_{p_h}(v_{ijk})>r_{ijk}$. As $D_h\in \calD'(i,j)$, $\{D_h,D_i,D_j\}$ is a strongly convex independent set. Therefore, $v_{ijh}$ exists by Lemma~\ref{lem:3IS}. Since $d_{p_h}(v_{ijk})>r_{ijk}=d_{p_i}(v_{ijk})$, applying Lemma~\ref{lem:preorder} (with $v=v_{ijk}$) gives $v_{ijh}\prec v_{ijk}$. 

    On the other hand, since $v_{ijk}\prec v_{ijk'}$, we have $v_{ijh}\prec v_{ijk'}$. Consequently, applying Lemma~\ref{lem:preorder}  (with $v=v_{ijk'}$) gives $d_{p_h}(v_{ijk'})>d_{p_i}(v_{ijk'})=r_{ijk'}$. This proves the lemma. 
\end{proof}

In light of Lemma~\ref{lem:monotone}, our algorithm proceeds as follows. We sort all indices $k \in K$ according to the $\prec$ order on $\gamma_{ij}$; let $k_1,k_2,\ldots,k_g$ be the resulting sequence, where $g=|K|$. Since $g\leq n$, this sorting step takes $O(n\log n)$ time. Next, for each disk $D_h \in \calD'(i,j)$, we compute the smallest index $t \in [1,g]$ such that $D_h$ is disjoint from $D_{ijk_t}$. By Lemma~\ref{lem:monotone}, this can be done in $O(\log n)$ time via binary search on the sequence $k_1,k_2,\ldots,k_g$. We then assign this value $t$ as a \emph{label} to $D_h$.

Recall that our goal is to compute $f(i,j,k_t)$ for all $t=1,2,\ldots,g$, where each $f(i,j,k_t)$ is equal to the maximum cost among all disks disjoint from $D_{ijk_t}$. We process the indices in increasing order of $t$. For convenience, let $f(i,j,k_0)=0$. For each $t\geq 1$, we initialize $f(i,j,k_t)=f(i,j,k_{t-1})$, and then examine all disks $D_l \in \calD'(i,j)$ whose label is exactly $t$. For each such disk, if $f(i,j,k_t)<\mathit{cost}(D_l)$, we update $f(i,j,k_t)=\mathit{cost}(D_l)$. In this manner, all values $f(i,j,k_t)$ for $t=1,2,\ldots,g$ can be computed in a total of $O(n)$ time.

This completes the processing for a fixed pair $(i,j)$, with a total running time of $O(n\log n)$. Since there are $O(n^2)$ choices of pairs $(i,j)$, the overall running time of the algorithm is $O(n^3\log n)$. The following lemma summarizes our result.

\begin{lemma}
\label{lem:strconvexindset}
Given a set $\calD$ of $n$ weighted disks in the plane such that the disks are in strongly convex position and $\calD$ admits a maximum-weight independent set that is strongly convex, a maximum-weight independent set of $\calD$ can be computed in $O(n^3\log n)$ time.
\end{lemma}

\section{The strongly-convex-input case}
\label{sec:input}

In this section, we consider the strongly-convex-input case, in which the input disks of $\calD$ are in strongly convex position but $\calD$ does not necessarily admit a maximum-weight independent set that is strongly convex (see Figure~\ref{fig:notstrcnv}).

\begin{figure}
\begin{minipage}[h]{\textwidth}
\begin{center}
\includegraphics[height=2in]{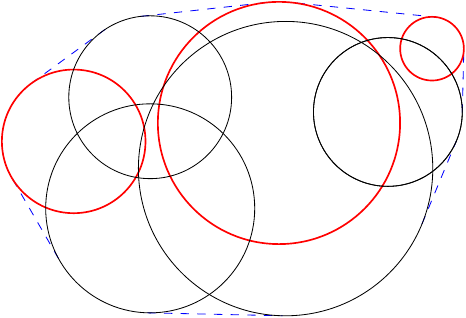}
\caption{\footnotesize Illustrating an example where the input set of disks is strongly convex, but the only maximum-weight independent set (shown in red) is not strongly convex (assuming the weight of each disk is $1$).}
\label{fig:notstrcnv}
\end{center}
\end{minipage}
\vspace{-0.1in}
\end{figure}

To handle this case, we reduce it to the setting in Section~\ref{sec:inputoutput}. To this end, we construct a set $B$ of $O(n)$ \emph{auxiliary points} satisfying the following two properties:
\begin{enumerate}
    \item No point of $B$ is contained in any disk of $\calD$.
    \item For any subset $\calD' \subseteq \calD$, $\calD' \cup B$ is strongly convex (in particular, $\calD \cup B$ is strongly convex).
\end{enumerate}

Each point of $B$ is treated as a special disk and assigned a weight of $1$. Let $S=\calD \cup B$. By the first property, any maximum-weight independent set of $S$ must include all points of $B$. Therefore, if $S'$ is a maximum-weight independent set of $S$, then $S' \setminus B$ is a maximum-weight independent set of $\calD$. In this way, computing a maximum-weight independent set of $\calD$ reduces to computing a maximum-weight independent set of $S$. Moreover, by the second property, $S$ is strongly convex and admits a maximum-weight independent set that is strongly convex. Consequently, computing a maximum-weight independent set of $S$ becomes an instance of the case in Section~\ref{sec:inputoutput} and can be solved in $O(n^3\log n)$ time by Lemma~\ref{lem:strconvexindset}.

In the remainder of this section, we focus on constructing a set $B$ of $O(n)$ auxiliary points satisfying the above two properties, and show that such a set can be computed in $O(n\log n)$ time.

To simplify the discussion, we assume that each disk in $\calD$ has positive radius. The same approach extends to the case where some disks degenerate to points, although the exposition becomes more involved. We also assume that $n \geq 3$, since smaller instances can be handled directly in $O(1)$ time. Because $\calD$ is strongly convex, each disk contributes exactly one arc to the boundary $\partial \calh(\calD)$. As in Section~\ref{sec:inputoutput}, let $\calD=\langle D_1,\ldots,D_n\rangle$ denote the cyclic order of disks appearing on $\partial \calh(\calD)$ in counterclockwise order.

\subsection{Defining $\boldsymbol{B}$}

For each disk $D_i \in \calD$, $D_i$ contributes a single arc $\alpha_i$ to the boundary $\partial \calh(\calD)$. Let $a_i$ denote the midpoint of $\alpha_i$, and let $\ell_i$ be the line through $a_i$ that is tangent to $D_i$. Thus, $\ell_i$ is also tangent to $\calh(\calD)$ at $a_i$. We orient $\ell_i$ so that $D_i$ lies to the left of $\ell_i$ (see Figure~\ref{fig:aug1A00}).

\begin{figure}[t]
\begin{minipage}[t]{\textwidth}
\begin{center}
\includegraphics[height=2.9in]{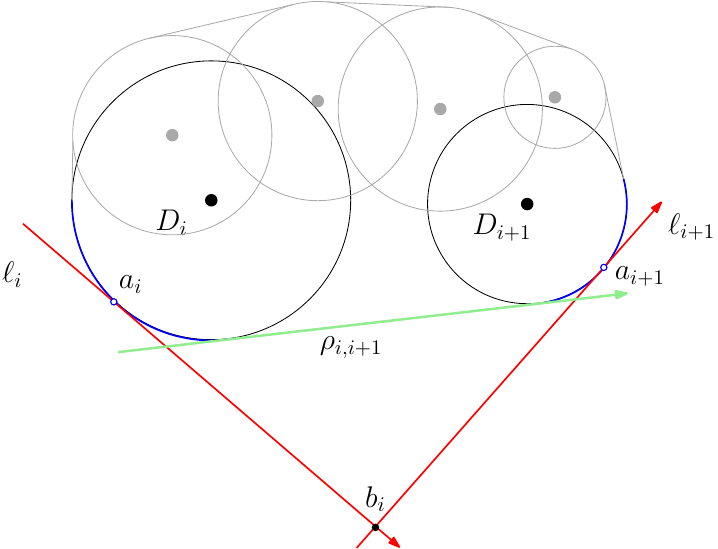}
\caption{\footnotesize Illustrating the definition of $b_i$.
}
\label{fig:aug1A00}
\end{center}
\end{minipage}
%\vspace{-0.1in}
\end{figure}

For each pair of consecutive disks $D_i$ and $D_{i+1}$ along $\partial \calh(\calD)$ (if $i=n$, then we let $i+1$ take modulo $n$, which equals $1$), we proceed as follows. Let $\rho_{i,i+1}$ be the directed common tangent line of $D_i$ and $D_{i+1}$ such that both disks lie to the left of $\rho_{i,i+1}$. By the definition of the disk order along $\partial \calh(\calD)$, the convex hull $\calh(\calD)$ lies entirely to the left of $\rho_{i,i+1}$, and the segment of $\rho_{i,i+1}$ between the two tangency points on $D_i$ and $D_{i+1}$ forms an edge of $\calh(\calD)$.

Since $\ell_i$ is tangent to $D_i$ at the midpoint of $\alpha_i$, the line $\ell_i$ must intersect $\rho_{i,i+1}$. Similarly, $\ell_{i+1}$ intersects $\rho_{i,i+1}$. Moreover, because $n \geq 3$, the lines $\ell_i$ and $\ell_{i+1}$ intersect at a point $b_i$ that lies to the right of $\rho_{i,i+1}$ (see Figure~\ref{fig:aug1A00}).

We define $B$ as the set of all such intersection points $b_i$, for $1 \leq i \leq n$. Thus, $|B| = n$.

If we treat each point $b_i$ as a special disk, then this construction places $b_{i-1}$, $D_i$, and $b_i$ in a degenerate configuration, since $\ell_i$ is tangent to all three. Note that our algorithm in Section~\ref{sec:inputoutput} also works for such degenerate cases. Alternatively, if desired, one could perturb each point $b_i$ infinitesimally toward $\rho_{i,i+1}$ so that the set $S=\calD \cup B$ is in general position. In what follows, we retain the definition of $B$ above and allow degeneracies, as this simplifies the exposition.

Since the convex hull $\calh(\calD)$ can be computed in $O(n\log n)$ time~\cite{ref:RappaportA92}, the set $B$ can also be constructed in $O(n\log n)$ time.

\subsection{Correctness}

We now prove that $B$ has the two desired properties. By definition, every point $b_i$ of $B$ is outside $\calh(\calD)$ and thus is outside every disk of $\calD$. Hence, the first property holds. In what follows, we focus on proving that $B$ has the second property. Let $S=\calD\cup B$. 

For each disk $D_i\in \calD$, let $H_i$ denote the closed halfplane bounded by $\ell_i$ and on the left of $\ell_i$. Hence, $H_i$ contains $\calh(\calD)$. Let $Q$ denote the common intersection of $H_i$ for all $1\leq i\leq n$ (see Figure~\ref{fig:aug1A10}). Note that $\calh(\calD)\subseteq Q$ as each halfplane $H_i$ contains $\calh(\calD)$. Furthermore, each line $\ell_i$ contains an edge of $Q$. Indeed, since $a_i\in  \calh(\calD)$ and $\calh(\calD)\subseteq H_j$ for each $j\neq i$, we know that $a_i$ is in all halfplanes $H_j$ with $1\leq j\leq n$ and $j\neq i$. Since $a_i$ is on the boundary of $H_i$, it follows that $a_i\in \partial Q$. Hence, $\ell_i$ contain a point of $\partial Q$. In addition, since the slopes of $\ell_i$, $i=1,2,\ldots,n$, in this order are monotonically changing, the intersections $b_i=\ell_i\cap \ell_{i+1}$ for all $i=1,2,\ldots,n$ are exactly the vertices of $Q$ in counterclockwise order along $\partial Q$ (and the line segments $\overline{b_{i}b_{i+1}}$ for all $i=1,2,\ldots,n$ are exactly the edges of $Q$ counterclockwise order along $\partial Q$).
Since $\calh(\calD)\subseteq Q$, $Q$ is exactly the convex hull $\calh(S)$ of $S$. 
For each $i$, as $a_i\in \partial Q$, we have $a_i\in \partial \calh(S)$.

\begin{figure}[t]
\begin{minipage}[t]{\textwidth}
\begin{center}
\includegraphics[height=2.9in]{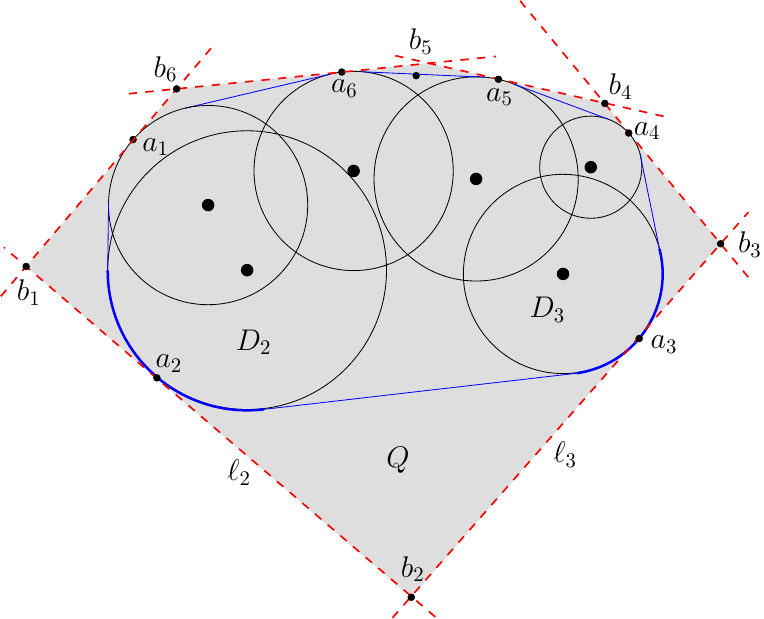}
\caption{\footnotesize Illustrating the definition of $Q$ (the gray region).
}
\label{fig:aug1A10}
\end{center}
\end{minipage}
%\vspace{-0.1in}
\end{figure}

Note that $a_i$ is in the only point of $D_i$ appearing on $\partial Q$. To see this, since $\calD$ is strongly convex, $D_i$ is in the interior of the halfplane $H_j$ for all $1\leq j\leq n$ with $j\neq i$. On the other hand, all points of $D_i$ other than $a_i$ are in the interior of $H_i$. Consequently, all points of $D_i$ except $a_i$ are in the interior of $Q$. 
Therefore, $a_i$ is in the only point of $D_i$ appearing on $\partial Q$.

With the above discussions, we are now in a position to prove the second property of $B$. 
Consider a subset $\calD'\subseteq \calD$. Our goal is to show that $S'=\calD'\cup B$ is strongly convex.
To this end, we will prove that every disk of $S'$ has only a single point appearing on $\partial \calh(S')$. 

We first show that $S'$ is convex. Indeed, for each auxiliary point $b_i\in B$, as $b_i\in \partial \calh(S)$ and $S'\subseteq S$,  $b_i$ must be on $\partial \calh(S')$. For each disk $D_i\in \calD'$, as $a_i\in D_i$, $a_i\in \partial \calh(S)$, and $S'\subseteq S$, we have $a_i\in \partial \calh(S')$. Thus, each disk of $S'$ must have a point on $\partial \calh(S')$. This proves that $S'$ is convex. 

We next show that for each disk $D_i\in \calD'$, $a_i$ is the only point appearing on $\partial\calh(S')$. Indeed, by definition,  $Q$ is also the convex hull $\calh(B)$ of $B$. Since $a_i$ is the only point of $D_i$ appearing on $\partial Q$, $a_i$ is the only point of $D_i$ appearing on $\partial \calh(B)$. Because $B\subseteq S'$ and $a_i\in \partial \calh(S')$, it follows that $a_i$ is the only point of $D_i$ appearing on $\partial \calh(S')$. 
This proves that $S'$ is strongly convex. 

In summary, we conclude that $B$ has the two desired properties. 

For reference purpose, the following lemma summarizes our result in this section. 

\begin{lemma}
    \label{lem:input}
Given a set $\calD$ of $n$ weighted disks in strongly convex position in the plane, a maximum-weight independent set of $\calD$ can be computed in $O(n^3\log n)$ time. 
\end{lemma}

\section{The general case}
\label{sec:general}

In this section, we finally consider the general case in which $\calD$ is not necessarily strongly convex (but is still convex). 
In particular, a disk may contribute more than one (maximal) arc to $\partial\calh(\calD)$.

To solve the problem, we reduce it to the case in Section~\ref{sec:input}. To this end, we create a set $Z$ of $O(n)$ {\em auxiliary points} as special disks with each disk assigned a weight equal to $1$ such that the following two properties hold: 
\begin{enumerate}
    \item No point of $Z$ is contained in any disk of $\calD$. 
    \item $\calD\cup Z$ is strongly convex.
\end{enumerate}

Let $S=\calD \cup Z$. As in the reduction in Section~\ref{sec:input}, due to the first property, it suffices to compute a maximum-weight independent set in $S$. This is an instance of the problem in Section~\ref{sec:input} due to the second property, and thus the problem can be solved in $O(n^3\log n)$ time by Lemma~\ref{lem:input}.

In the rest of this section, we will focus on defining a set $Z$ of $O(n)$ points satisfying the above two properties. We will also show that computing $Z$ can be done in $O(n\log n)$ time. 

Consider a disk $D_i$ that has more than one arc on $\partial \calh(\calD)$. 
Note that at most one such arc can have a radian measure larger than or equal to $\pi$, and all other arcs have median measures strictly smaller than $\pi$. Let $\alpha_i$ be an arc of $D_i$ on $\partial \calh(\calD)$ whose radian measure is less than $\pi$.

Let $\alpha_{i-1}$ and $\alpha_{i+1}$ be the two disk arcs on $\calh(\calD)$ adjacent to $\alpha_i$ clockwise and counterclockwise, respectively; let $D_{i-1}$ and $D_{i+1}$ be the disks where these arcs lie on, respectively (see Figure~\ref{fig:aug2new}). 
Note that it is possible that $D_{i-1}$ (resp., $D_{i+1}$) is a single point. We first assume that both $D_{i-1}$ and $D_{i+1}$ have radii larger than $0$. We will discuss the other case later, which can be handled similarly. Refer to Figure~\ref{fig:aug2new} for an illustration of the notation introduced below. 

For each $j\in \{i-1,i,i+1\}$, let $x_j$ and $y_j$ be the clockwise and counterclockwise endpoints of $\alpha_j$, respectively. By definition, $\overline{y_{i-1}x_i}$ is an edge of $\calh(\calD)$ and its supporting line is tangent to $D_{i-1}$ at $y_{i-1}$ and tangent to $D_i$ at $x_i$. Let $\rho_{i-1,i}$ denote the supporting line of $\overline{y_{i-1}x_i}$. We orient $\rho_{i-1,i}$ so that $\calh(\calD)$ is on the left side of $\rho_{i-1,i}$. Similarly, let $\rho_{i,i+1}$ be the oriented supporting line of $\overline{y_ix_{i+1}}$ so that $\calh(\calD)$ is on its left side. 

Let $\ell(\overline{x_iy_i})$ denote the line containing $x_i$ and $y_i$. 
Without loss of generality, we assume that $\ell(\overline{x_iy_i})$ is horizontal such that $x_i$ is left of $y_i$. Thus, the arc $\alpha_i$ is below $\ell(\overline{x_iy_i})$.
Since the radian measure of $\alpha_i$ is less than $\pi$, $\rho_{i-1,i}$ must intersect $\rho_{i,i+1}$ at a point $q$ strictly below $\ell(\overline{x_iy_i})$. 

\begin{figure}[t]
\begin{minipage}[h]{\textwidth}
\begin{center}
\includegraphics[height=3.6in]{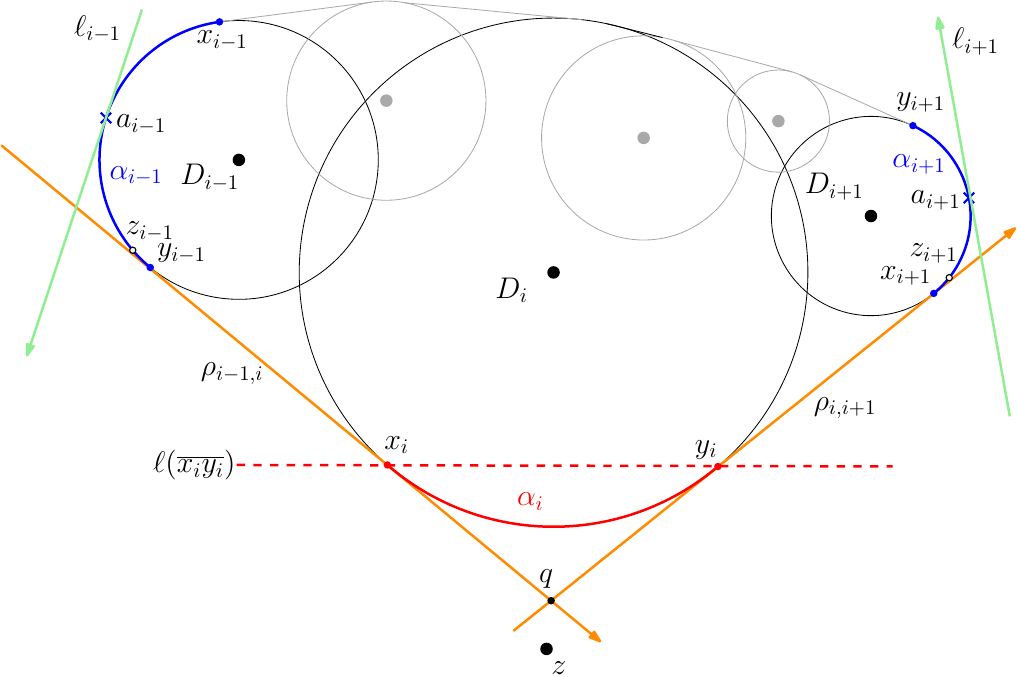}
\caption{\footnotesize Illustration of the definition of $z$.}
\label{fig:aug2new}
\end{center}
\end{minipage}
\vspace{-0.1in}
\end{figure}

Let $a_{i-1}$ be the middle point of the arc $\alpha_{i-1}$. Let $\ell_{i-1}$ be an oriented line tangent to $D_{i-1}$ at $a_{i-1}$ such that $\calh(\calD)$ is on the left of $\ell_{i-1}$. Similarly, we define $a_{i+1}$ and $\ell_{i+1}$ for $D_{i+1}$.

Note that by definition, $q$ must be to the left of both $\ell_{i-1}$ and $\ell_{i+1}$. Let $q'$ be a point strictly below $q$ but infinitesimally close to $q$. Then, $q'$ is strictly to the right of both $\rho_{i-1,i}$ and $\rho_{i,i+1}$, and also strictly to the left of both $\ell_{i-1}$ and $\ell_{i+1}$. This means that the region $Q$ is not empty, where $Q$ is the common intersection of the following four open halfplanes: the open halfplane to the right of $\rho_{i-1,i}$, the open halfplane to the right of $\rho_{i,i+1}$, the open halfplane to the left of $\ell_{i-1}$, and the open halfplane to the left of $\ell_{i+1}$. 

We pick an arbitrary point $z$ from $Q$ and use it as an auxiliary point; we show below that $z$ can be used to eliminate the arc $\alpha_i$. 

Let $z_{i-1}$ be a point on $\partial D_{i-1}$ such that $\overline{z_{i-1}z}$ is tangent to $D_{i-1}$ at $z_{i-1}$ and $D_{i-1}$ is to the left of the directed line containing $\overline{z_{i-1}z}$ with the direction from $z_{i-1}$ to $z$. Since $z$ is strictly to the left of $\ell_{i-1}$ and strictly to the right of $\rho_{i-1,i}$, the tangent point $z_{i-1}$ must be in the interior of the arc $\partial_{[a_{i-1},y_{i-1}]}D_{i-1}$, which is a subarc of $\alpha_{i-1}$. 

Symmetrically, let $z_{i+1}$ be a point on $\partial D_{i+1}$ such that $\overline{z_{i+1}z}$ is tangent to $D_{i+1}$ at $z_{i+1}$ and $D_{i+1}$ is to the right of the directed line containing $\overline{z_{i+1}z}$ with the direction from $z_{i+1}$ to $z$. Since $z$ is strictly to the left of $\ell_{i+1}$ and strictly to the right of $\rho_{i,i+1}$, the tangent point $z_{i+1}$ must be in the interior of the arc $\partial_{[x_{i+1},a_{i+1}]}D_{i+1}$, which is a subarc of $\alpha_{i+1}$.

Consider the convex hull $\calh(\calD\cup \{z\})$. Its boundary can be obtained from $\partial \calh(\calD)$ by replacing $\partial_{[z_{i-1},z_{i+1}]}\calh(\calD)$ with $\overline{z_{i-1}z}\cup \overline{zz_{i+1}}$. 
In particular, no point of $\alpha_i$ appears on $\partial \calh(\calD\cup \{z\})$.
Hence, every disk of $\calD$ has the same number of arcs appearing in $\partial \calh(\calD\cup \{z\})$ as in $\partial \calh(\calD)$ except that $D_i$ has one less arc than before. Furthermore, since $z$ is vertex of $\calh(\calD\cup \{z\})$, $\calD\cup \{z\}$ is still convex. 

The above assumes that both $D_{i-1}$ and $D_{i+1}$ are disks of radii larger than zero. The other case can be handled similarly. For example, if $D_{i-1}$ is a point $p_{i-1}$, then following the same way as above, we have $\alpha_{i-1}=x_{i-1}=y_{i-1}=a_{i-1}=z_{i-1}=p_{i-1}$. 
Hence, $p_{i-1}$ is still on $\partial \calh(\calD\cup \{z\})$. Therefore, it is still the case that $\calD\cup \{z\}$ is convex and  every disk of $\calD$ has the same number of arcs appearing in $\partial \calh(\calD\cup \{z\})$ as in $\partial \calh(\calD)$ except that $D_i$ has one less arc than before. If $D_{i+1}$ is a point, we can handle the situation similarly. 

In summary, we can find a point $z$ such that $\calD\cup \{z\}$ is convex and every disk of $\calD$ has the same number of arcs appearing in $\partial \calh(\calD\cup \{z\})$ as in $\partial \calh(\calD)$ except that $D_i$ has one less arc than before.  
We add $z$ to $Z$. Note that $z$ is outside every disk of $\calD$.

If $\calD\cup \{z\}$ is not strongly convex, then by following the same method we can find another auxiliary point with respect to $\calD\cup \{z\}$ to eliminate another arc. Since $\partial \calh(\calD)$ has $O(n)$ disk arcs~\cite{ref:RappaportA92}, we can find a set $Z$ of $O(n)$ auxiliary points such that $\calD\cup Z$ is strongly convex (and no point of $Z$ is inside any disk of $\calD$). 
After $\calh(\calD)$ is constructed initially, which can be done in $O(n\log n)$ time~\cite{ref:RappaportA92}, computing each auxiliary point $z$ takes $O(1)$ time following the above method. Hence, computing $Z$ can be done in $O(n\log n)$ time in total.

The following theorem summarizes our result. 

\begin{theorem}\label{theo:convexindset}
Given a set $\calD$ of $n$ weighted disks in convex position in the plane, a maximum-weight independent set of $\calD$ can be computed in $O(n^3\log n)$ time.
\end{theorem}

The following corollary will be used in Section~\ref{sec:dispersion} to solve a dispersion problem.

\begin{corollary}\label{coro:10}
Given a set $\calD$ of $n$ weighted disks in convex position in the plane and a parameter $r>0$, one can compute in $O(n^3\log n)$ time a maximum-cardinality subset $\calD'\subseteq \calD$ such that the distance between any two disks in $\calD'$ is greater than $r$.
\end{corollary}
\begin{proof}
For each disk $D\in \calD$, let $D(r)$ be the disk with the same center as $D$ and with radius $r/2$ larger than that of $D$. 
For any subset $\calD'\subseteq \calD$, define $S(\calD')=\{D(r): D\in \calD'\}$. It is not difficult to see that for any subset $\calD'\subseteq \calD$, the distance between any two disks in $\calD'$ is greater than $r$ if and only if $S(\calD')$ is an independent set. Hence, the problem reduces to finding a maximum-cardinality independent set from $S(\calD)$. Since $\calD$ is in convex position, $S(\calD)$ must also be in convex position because each disk of $\calD$ grows by the same radius to obtain $S(\calD)$. Therefore, if we assign each disk of $S(\calD)$ a weight equal to $1$, then we can compute a maximum-cardinality independent set of $S(\calD)$ in $O(n^3\log n)$ time by Theorem~\ref{theo:convexindset}. 
\end{proof}

\section{The dispersion problem}
\label{sec:dispersion}

Let $\calD$ be a set of $n$ disks in convex position in the plane and let $k>0$ be an integer. For any two disks $D,D'\in \calD$, we define the {\em disk distance}  $\dist(D,D')$ as the minimum distance between any two points $q\in D$ and $q'\in D'$. 
The \textit{dispersion problem} for $\calD$ asks for a subset $\calD'\subseteq \calD$ of $k$ disks maximizing the minimum pairwise distance $\dist(D,D')$ of any two disks $D,D'\in \calD'$. 

Let $r^*$ denote the optimal objective value. It is not difficult to see that $r^*$ is equal to $\dist(D,D')$ for two disks $D,D'\in \calD$. Hence, $r^*$ belongs to the set $R=\{\dist(D,D'): D, D'\in \calD \}$. Clearly, $|R|=O(n^2)$.

Given a number $r\ge 0$, the \textit{decision problem} is to determine whether $r< r^*$, or equivalently, whether $\calD$ has a subset of $k$ disks whose minimum pairwise distance is greater than $r$. By Corollary~\ref{coro:10}, the decision problem can be solved in $O(n^3\log n)$ time.
To find $r^*$, we first compute $R$ and then do binary search in $R$ using the decision algorithm.
This yields an $O(n^3\log^2 n)$-time algorithm to compute $r^*$.
We can also find an optimal subset by Corollary~\ref{coro:10}, as explained in the proof of the following theorem. 

\begin{theorem}
Given a set of $n$ disks in convex position in the plane and an integer $k$, one can find a subset of $k$ disks maximizing their minimum pairwise distance in $O(n^3\log^2 n)$ time.
\end{theorem}

\begin{proof}
We first compute $r^*$ as discussed above. Once $r^*$ is computed, we apply the algorithm of Corollary~\ref{coro:10} to produce an optimal subset with $r=r^*$. The algorithm of Corollary~\ref{coro:10} will first compute a set $S(\calD)$ of disks and then apply the algorithm of Theorem~\ref{theo:convexindset} on $S(\calD)$. Here we need to slightly modify the algorithm of Theorem~\ref{theo:convexindset} by treating each disk of $S(\calD)$ as an open disk. For example, two disks are considered to be disjoint if they are outer tangent to each other. 

\end{proof}

% \section*{Declarations}
% \label{sec:Declarations}

% \paragraph{Conflict of interest statement.}
% The authors declare that they have no known competing financial interests or personal relationships that could have appeared to influence the work reported in this paper.

%\footnotesize

\bibliographystyle{plainurl}
\bibliography{refs}
\end{document}